\documentclass[11pt]{article}
\usepackage{amsfonts,latexsym,amsthm,amssymb,amsmath,amscd,euscript,tikz, bm, mathtools, commath,centernot,float}
\usepackage{algorithm}
\usepackage{algpseudocode}
\usepackage{fullpage}
\usepackage{amsmath}
\usepackage{dsfont}
\usepackage{enumerate}
\usepackage{hyperref}
\usepackage{bbm}
\usepackage{comment}
\usepackage{setspace}
\usepackage[authoryear,round]{natbib}
\usepackage{subcaption}
\usepackage{graphicx}
\usepackage[margin=1in]{geometry}

\hypersetup{colorlinks=true,citecolor=blue,urlcolor =black,linkbordercolor={1 0 0}}
\hypersetup{
  linkcolor=cyan 
}

\allowdisplaybreaks[1]

\newtheorem{theorem}{Theorem}
\newtheorem{proposition}{Proposition}
\newtheorem{lemma}{Lemma}
\newtheorem{corollary}{Corollary}

\theoremstyle{definition}

\theoremstyle{definition}

\newtheorem{example}[theorem]{Example}

\newcommand{\wh}{\widehat}
\newcommand{\indep}{\perp \!\!\! \perp}
\newcommand{\textmax}{\text{max}}
\newcommand{\textmin}{\text{min}}

\newcommand{\Prob}{\mathbb{P}}
\newtheorem{assumption}{Assumption}

\theoremstyle{remark}
\newtheorem{remark}{Remark}

\title{A sensitivity analysis for the average derivative effect} 
\date{}
\author{Jeffrey Zhang\\{\href{mailto:jeffzhang@uchicago.edu}{{\tt jeffzhang@uchicago.edu}}}}

\begin{document}

\maketitle

\begin{abstract}
In observational studies, exposures are often continuous rather than binary or discrete. At the same time, sensitivity analysis is an important tool that can help determine the robustness of a causal conclusion to a certain level of unmeasured confounding, which can never be ruled out in an observational study. Sensitivity analysis approaches for continuous exposures have now been proposed for several causal estimands. In this article, we focus on the average derivative effect (ADE). We obtain closed-form bounds for the ADE under a sensitivity model that constrains the odds ratio (at any two dose levels) between the latent and observed generalized propensity score. We propose flexible, efficient estimators for the bounds that are compatible with machine learning algorithms as well as point-wise and simultaneous (over the sensitivity parameter) confidence intervals. We examine the finite sample performance of the methods through simulations and illustrate the methods on a study assessing the effect of parental income on educational attainment and a study assessing the price elasticity of petrol.
\end{abstract}

\section{Introduction}
\label{section: intro}
Drawing causal inferences from observational studies is challenging for a variety of reasons. Two prominent challenges are a) the treatment (or exposure) of interest is continuous rather than binary or discrete, and b) the treatment/exposure assignment is not random, and so treated and control groups may substantially differ on the basis of unmeasured confounders. We are motivated by the intersection of these two challenges.

Continuous exposures are widespread in observational studies. In epidemiology, researchers are often interested in quantifying the effect of lead exposure or air pollution on health outcomes \citep{Dominici2006, Gibson2022}. In the social sciences, researchers may be interested in quantifying the effect of household income on a child's later life educational or economic outcomes \citep{Lundberg2023}. The vast majority of causal inference methods focus on treatments that are binary or take on a finite number of values. A popular approach when the exposure is continuous is to simply dichotomize the exposure at some threshold. However, this is undesirable because dichotomizing discards potentially useful statistical information from the exposure level and can invalidate the downstream inference \citep{VanderWeele2013}, or at least muddies its interpretation \citep{Lee2024}. When the exposure of interest is continuous, some popular causal estimands include the dose-response curve (also known as the exposure response function) \citep{kennedy_dose_response}, a (projection) parameter of the dose-response curve \citep{Bonvini2022SensitivityModels}, average derivative effects \citep{Newey1993EfficiencyModels, Klyne2023AverageLearning}, and stochastic intervention effects \citep{Schindl2024}, among others.

In this paper, we focus on the average derivative effect (ADE), which has been extensively studied in economics \citep{Hardle1989, Newey1993EfficiencyModels}, with recent work in statistics clarifying its causal interpretation \citep{Rothenhausler2019, Hines2021ParameterisingEffects}. It has many other names, and has been referred to as an incremental effect \citep{Rothenhausler2019}, average partial effect \citep{Klyne2023AverageLearning}, or an average causal derivative \citep{Chernozhukov2022LongLearning}. Roughly speaking, the ADE measures the average effect of infinitesimally increasing the exposure level for all units in the population. The ADE corresponds to well-known quantities in economic applications. For example, when one is interested in the effect of the price of a good on demand, the ADE corresponds to the average price elasticity of demand. In cases where one is interested in the effect of increased disposable income on consumption, the ADE corresponds to the average marginal propensity to consume \citep{bruns2025two}. As elucidated in \cite{Hines2021ParameterisingEffects}, the ADE estimand has some attractive properties, relative to the widely studied dose-response curve. For example, the ADE is a single number summary of the causal effect. Relatedly, estimation of the ADE, a scalar, is a simpler statistical task than estimating the dose-response curve, which is infinite-dimensional. The ADE also relies on a weaker version of the overlap/positivity assumption than that required for identification of the dose-response curve. Finally, as \cite{Hines2021ParameterisingEffects} explain, it may be difficult to envision an intervention that sets exposure to the same level for everyone, which is what the dose-response curve measures. Of course, there are settings where the ADE is less appropriate than the dose-response curve. Notably, if the causal effect is not monotonic, an average of positive and negative derivatives could cancel and result in a near zero ADE, rendering the ADE an inadequate summary of the causal effect. Data where the treatment is continuous almost exclusively arises in observational studies, as continuous treatments are rare in randomized experiments. As a result, the ADE is typically only a relevant estimand in observational studies, where unmeasured confounding can never be ruled out. It is therefore of interest to develop methods for ADE estimation that take into account potential unmeasured confounding.

A popular way to alleviate concerns about unmeasured confounding in an observational study is to perform a sensitivity analysis. Dating back to \citet{Cornfield2009}, a sensitivity analysis acknowledges the existence of unmeasured confounding, but asks how strong it must be to overturn a qualitative causal conclusion. There are now a wide array of sensitivity analysis methods, especially for binary exposures. Different sensitivity models focus on different aspects, such as the treatment assignment \citep{rosenbaum_obs, Tan2006AScores}, the potential outcome distribution, \citep{Robins2000, Diaz2013, Nabi2024}, worst-case departures \citep{Yadlowsky2022,Dorn2023}, average-case departures \citep{Huang2024, Zhang2022a}, and the proportion of confounding \citep{Bonvini2022}. There is a growing body of literature on sensitivity analyses for non-binary treatments. For the multi-valued treatment case, \cite{Basit2023SensitivityTreatments} derive bounds on linear combinations of potential outcome means. \cite{zheng2025copula} consider multivariate treatments under a linear structural model. For the continuous treatment case, most sensitivity analysis methods have been aimed at bounds on (aspects of) the dose-response curve/exposure response function. For example, \cite{Bonvini2022SensitivityModels} propose sensitivity analyses for parameters of marginal structural models. \cite{Jesson2022}, \cite{Marmarelis2023} \cite{Frauen2023}, and \cite{Baitairian2024} all propose distinct sensitivity analyses/models for the dose-response curve. One of the models considered in \cite{dalal2025partial} is closely related to ours, but the estimand of interest is the dose-response curve. The recent work of \cite{Levis2024} focuses on stochastic intervention estimands. Meanwhile, \cite{Zhang2024a} and \cite{Zhang2024b} propose sensitivity analysis methods for matched studies with continuous treatments under a related sensitivity model.

\cite{Chernozhukov2022LongLearning} study omitted variable bias within a general class of estimands that are continuous, linear functionals of a conditional expectation, with the ADE being one example. They express the omitted variable bias in terms of $L_2$ or $R^2$-style measures based on deviations between the ``short'' and ``long'' Riesz representers and conditional expectation functions. We take a substantially different approach by appealing to a sensitivity model that imposes a bound on the $L_\infty$ distance between the odds ratio of the observed and unobserved generalized propensity score at any pair of points in the support of the exposure. Notably, our model is closely related to widely used sensitivity models for binary treatments, facilitating a familiar interpretation. This gives rise to both different interpretations of the sensitivity parameters and different bounds on the ADE. Moreover, we derive tailored bounds for the binary outcome case; incorporating knowledge of a binary outcome appears to be challenging in the framework of \cite{Chernozhukov2022LongLearning}. We give a more detailed comparison to previous sensitivity models in Remarks \ref{remark: connection to other continuous} and \ref{remark: connection to long story short} in Section \ref{section: sensitivity model}.

Our contributions are as follows. For the ADE, we introduce a new sensitivity analysis model that bounds worst-case departures of the treatment assignment from no unmeasured confounding (Sections \ref{section: prelim} and \ref{section: sensitivity model}). An advantage of the model is that it is closely related to the familiar marginal sensitivity model for binary treatments due to \cite{Tan2006AScores}, as well as to the model of \cite{rosenbaum1989sensitivity}, which itself is a generalization of the popular Rosenbaum sensitivity model for binary treatments \citep{rosenbaum_obs}. Under the sensitivity model, we derive closed-form bounds for the ADE for \emph{both} continuous and binary outcomes (Section \ref{section: optimization}). In the semiparametric efficiency theory/debiased machine learning tradition, we introduce nonparametric, robust estimators for the bounds, as well as corresponding confidence intervals (Section \ref{section: estimation and inference}). Importantly, these estimators can leverage data-adaptive or machine learning nuisance estimators that may converge slower than the parametric rate (though not too slowly). An especially attractive property of our approach is the ease of conducting simultaneous inference over an interval of the sensitivity parameter. We demonstrate the practicality of the bounds through simulation and two data analyses (Sections \ref{section: sims} and \ref{section: application}). Code to reproduce the simulations and data analyses is available at \url{https://github.com/jzhang1937/ade-sens}. 

\section{Preliminaries}
\label{section: prelim}
In this section, we formally introduce notation, assumptions, and the average derivative effect.
\subsection{Notation}
The underlying data are assumed to be a vector of independent and identically distributed samples from some unknown distribution $(A,X,U,Y(a)_{a \in \mathcal{A}}) \sim F$, where $A$ is a continuous exposure (also referred to as a treatment or dose), $X$ are observed pre-exposure confounders, $U$ some unobserved pre-exposure confounders, $Y(a)$ are potential outcomes, and $\mathcal{A}$ is the support of the exposure. The potential outcomes represent the outcome quantity that would have been observed if the exposure had been externally set to $a$. In the observed data, we only see one of the potential outcomes. Specifically, the observed outcome satisfies $Y = Y(a)$ when $A = a$ (Assumption \ref{assumption: SUTVA}). Thus, the observed data consists of $n$ i.i.d. samples of the data vector $(X, A, Y)$. Throughout the paper, we will consider the two scenarios where $Y$ is continuous and $Y$ is binary separately. We will use $'$ and $\partial_a$ to denote taking a partial derivative with respect to exposure $a$, which will be clear from the context. $\norm{\cdot}$ represents the $L_2$ norm, $\indep$ denotes statistical independence, and $\mathbbm{1}$ denotes the indicator function.

\subsection{Estimand and assumptions}
We now formally introduce the causal estimand, which is well-defined for both binary and continuous (potential) outcomes.  The ADE is defined as
\begin{equation}
\label{eq: ADE definition}
\theta \equiv \lim_{\delta \to 0} \delta^{-1} E[Y(A+\delta)-Y(A)].
\end{equation}
It is useful to unpack the meaning of $\theta$. For some small $\delta > 0$, note that $E[Y(A+\delta)-Y(A)]$ captures an average difference in outcomes in two worlds. In the first world, every unit's observed exposure is increased by $\delta$ from its natural value \citep{haneuse2013modified}. In the second, there is no intervention on the exposure, i.e. exposure takes its natural value and $E[Y(A)] = E[Y]$. $\theta$ then captures the limiting difference between average outcomes under the small $\delta$ shift intervention and no intervention, scaled by the shift \citep{hines2025learning}. When $Y$ is continuous, $\theta$ also matches the $E[Y'(A)]$ estimand introduced in \cite{Rothenhausler2019} under additional smoothness conditions that ensure $E[Y'(A)]$ is well-defined. The required regularity conditions are made explicit in Appendix \ref{sec: estimand appendix}. We now introduce causal assumptions that allow the ADE to be written as a statistical functional of the full data, some components of which are not observed. 

\begin{assumption}[Consistency/SUTVA]
\label{assumption: SUTVA}
$Y = Y(a)$ when $A = a$.
\end{assumption}
Consistency requires that the potential outcome corresponding to an exposure $a$ matches the observed outcome when the observed exposure matches $a$. It also prohibits interference, where the exposure of one individual can affect the outcome of another.

\begin{assumption}[Latent Ignorability]
\label{assumption: latent ignorability}
$\{Y(a)\}_{a \in \mathcal{A}} \indep A \mid X, U$.
\end{assumption}

Latent ignorability essentially requires that the $(X, U)$ vector contains all common causes of the exposure $A$ and outcome $Y$. This assumption is notably weaker than the typical ignorability or no unmeasured confounding assumption, which requires $\{Y(a)\}_{a \in \mathcal{A}} \indep A \mid X$. Moreover, it is a relatively mild assumption since $U$ is unobserved, and there are no restrictions placed on $U$.

\begin{assumption}[Local Overlap]
\label{assumption: local overlap}
The conditional densities $f(a \mid x)$ and $f(a \mid x, u)$ are continuous in $a$ for all $x, u$.
\end{assumption}
In contrast to a typical overlap assumption that would require $f(a \mid x, u) > 0$ for all $a \in \mathcal{A}$ (this would be needed for the dose-response curve), the local overlap assumption only requires that at any level of confounders $(u, x)$ for which $f(a \mid x, u) > 0$, the conditional density must be positive in a small neighborhood around $a$ as well. We only require this weaker notion of overlap as we focus on the ADE. We will also refer to the conditional density $f(a \mid x, u)$ as the generalized propensity score throughout the remainder of the paper \citep{Imbens2000}. The next assumption collects additional regularity conditions on the conditional densities and potential outcomes.
\begin{assumption}[Regularity]
\label{assumption: regularity}
\leavevmode
\par\vspace{0.01em}
\begin{enumerate}[(i)]
\item The conditional densities $f(a \mid x)$ and $f(a \mid x, u)$ are differentiable in $a$. Also, $f(a \mid x) \to 0$ and $f(a \mid x, u) \to 0$ for all fixed $x, u$ as $|a| \to \infty$.
\item The derivative $\partial_a E[Y(a) \mid X = x, U = u] \equiv \lim_{\delta \to 0} \delta^{-1} E[Y(a+\delta) - Y(a) \mid x, u]$ exists and is bounded and continuous. 
\end{enumerate}
\end{assumption}

These regularity assumptions ensure that certain statistical objects are well-defined. In essence, they place smoothness restrictions on the potential outcomes and conditional densities.  

\begin{lemma}
\label{prop: binary estimand}
 Under Assumptions \ref{assumption: SUTVA}-\ref{assumption: regularity},
\begin{equation*}
    \theta = E[\partial_a E[Y(A) \mid X, U]] =  E[\partial_a E[Y \mid  A , X , U ]] = E[-s(A \mid X, U)Y],
\end{equation*}
where $ s(a\mid x, u)\equiv \frac{f'(a \mid x, u)}{f(a\mid x, u)}$, is the (conditional) score function.
\end{lemma}
Here, $-s(A\mid X, U)$ is the Riesz representer for the statistical functional $E[\partial_a E[Y \mid  A , X , U ]]$ \citep{Powell1989}. Two statistical examples where the interpretation of $\theta$ is fairly simple are the partially linear model and the single index model.
\begin{example}[Partially linear model]
Consider the case where $E[Y \mid A, X, U] = \beta A + g(X, U)$, for $g$ an arbitrary function and $\beta$ a scalar. Then $\theta = \beta$.
\end{example}
\begin{example}[Single index model \citep{Stoker1986}]
Suppose $E[Y \mid A, X, U] = F(\beta A + \beta_X X + \beta_U U)$, for $F$ a monotone, differentiable function. Then $\theta = \beta \times E[\partial_a F(\beta A + \beta_X X + \beta_U U)]$. So the $\theta$ estimand is proportional to the regression coefficient $\beta$.
\end{example}
For the remainder of the paper, we focus on the statistical functional $\theta$ in the nonparametric model. In Appendix \ref{sec: weighted ade}, we briefly outline how our results could be extended to weighted average derivative effects, i.e. $\theta_w \equiv E[w(A, X)\partial_a E[Y(A) \mid X, U]]$, for weights $w(A,X)$ that are nonnegative and such that $E[w(A, X)] = 1$. \cite{Hines2021ParameterisingEffects} and \cite{Hines2023OptimallyEffects} discuss the interpretation of such weighted average derivative effects.

\section{Sensitivity model}
\label{section: sensitivity model}
Based on the previous section, the causal estimand equals the statistical functional $\theta$. However, $\theta$ is a function of $E[Y \mid A, X, U]$, where $U$ is unobserved. Thus, $\theta$ cannot be estimated from observed data without further assumptions. In this section, we introduce a sensitivity analysis model that, through a parameter $\gamma$, limits the strength of association between the unmeasured confounder $U$ and the exposure $A$ conditional on $X$. The sensitivity model then facilitates deriving upper and lower bounds on the estimand of interest based on the allowable amount of unmeasured confounding $\gamma$. Recall $\theta$ can be expressed by $\theta = E[\partial_a E[Y \mid  A , X , U ]]$ or as a function of its Riesz representer, $\theta = E[-s(A \mid X, U)Y]$. The sensitivity model we consider is a generalization of the model of \cite{Tan2006AScores} to the continuous exposure case and is related to Rosenbaum's semiparametric model for continuous doses \citep{rosenbaum1989sensitivity}. The model restricts the odds ratio of the generalized propensity score (including $u$) vs. the generalized propensity score (marginalizing over $u$) at any two dose levels:
\begin{assumption}[Marginal $\gamma$ sensitivity model]
\begin{equation}
\label{eqn:sens_model}
    \exp(-\gamma(|a - a'|)) \leq \frac{f(a' \mid x, u)f(a\mid x)}{f(a\mid x, u)f(a' \mid x)}\leq  \exp(\gamma(|a - a'|)) \ \forall x, a, a', u.
\end{equation}
\end{assumption}
This model generalizes Tan's marginal sensitivity model for binary treatments, as it exactly reduces to the model of \cite{Tan2006AScores} when $a$ and $a'$ are replaced by 0 and 1, and the $\Gamma$ from \cite{Tan2006AScores} is set to $\exp(\gamma)$. The Tan model has been studied extensively \citep{Zhao2019SensitivityBootstrap, Dorn2023}, and has recently been generalized to longitudinal settings \citep{bruns2023robust, Tan2025}. \cite{dalal2025partial} consider a more general formulation to bound dose-response curves. The connection between the Tan and Rosenbaum models has been previously discussed when $A$ is binary and continuous \citep{Zhao2019SensitivityBootstrap, dalal2025partial}. We examine the connection between model \eqref{eqn:sens_model} and Rosenbaum's model for continuous doses specific to our context in Appendix \ref{sec: compare marginal and Rosenbaum}. In Appendix \ref{sec: alternative transform}, we discuss implications of instead imposing $g(-\gamma(|a-a'|)) \leq \frac{f(a' \mid x, u)f(a\mid x)}{f(a\mid x, u)f(a' \mid x)} \leq g(\gamma(|a-a'|)) \ \forall a, x, u$ for a smooth, nonnegative, strictly increasing function $g$ such that $g(0) = 1$. We now comment on the relationship between model \eqref{eqn:sens_model} in relation to other sensitivity models for continuous exposures.

\begin{remark}
\label{remark: connection to other continuous}
In relation to other sensitivity models for other causal estimands for continuous exposures, ours model differs in that we consider bounds on the odds deviation of generalized propensity scores, and the deviation depends on the difference between dose levels $a$ and $a'$. In contrast, \cite{Bonvini2022SensitivityModels}, \cite{Jesson2022}, and \cite{Baitairian2024} consider models that place the restriction $1/\Gamma \leq f(a \mid x, u)/ f(a \mid x) \leq \Gamma \ \forall a, x, u$. This model does not involve odds at different dose levels, and so the bounds do not depend on the dose difference. \cite{Marmarelis2023} and \cite{dalal2025partial} consider closely related models to ours, but do not study the ADE. One of the models considered by \cite{Levis2024} is similar to ours as it considers restrictions on odds, but does not involve the dose difference, and they consider a different estimand.
\end{remark}
\begin{remark}
\label{remark: connection to long story short}
\cite{Chernozhukov2022LongLearning} establishes bounds on the bias of estimators that ignore unmeasured confounding for a large class of functionals that can be expressed as a continuous linear function of the conditional expectation $E[Y \mid A, X, U]$. Their class includes the ADE from \eqref{eq: ADE definition}, or what they refer to as the average causal derivative. Their bounds are based on $L_2$ or $R^2$ measures comparing both the ``long'' (including $U$) Riesz representer and the conditional expectation relative to their ``short'' (excluding $U$) counterparts. The model we consider differs in several aspects. First, model \eqref{eqn:sens_model} imposes bounds on worst-case rather than average case departures from no unmeasured confounding, but only for the treatment. Therefore, the models and their interpretations are fundamentally different. Second, in contrast to \cite{Chernozhukov2022LongLearning}, we do not parametrize our sensitivity model in terms of the ``long'' Riesz representer $s(a \mid x, u)$ in relation to its ``short'' counterpart $s(a \mid x)$. Specifically, \cite{Chernozhukov2022LongLearning} consider $L_2^s = E[\{s(A \mid X, U) - s(A \mid X)\}^2]$ and $1 - R_s^2 \equiv (E[s(A \mid X, U)^2] - E[s(A \mid X)^2]) / E[s(A \mid X, U)^2] $. We impose restrictions on the generalized propensity score $f(a \mid x, u)$ directly, which may be a more familiar statistical object than $s(a \mid x, u)$, and may aid in the interpretation of the sensitivity parameter. Nevertheless, as will be shown in Lemma \ref{lemma: model implication on latent score}, our model \eqref{eqn:sens_model} implies constraints on  $s(a \mid x, u)$. As a result, the connection between \eqref{eqn:sens_model} and $s(a \mid x, u)$ derived in Lemma \ref{lemma: model implication on latent score} could be useful in interpreting the $L_2^s$ and $R_s^2$ measures examined in \cite{Chernozhukov2022LongLearning}. We also point out that Proposition 2 from the earlier work of \cite{Rothenhausler2019} derives a similar bound on the omitted variable bias as \cite{Chernozhukov2022LongLearning} do, specifically for the ADE.
\end{remark}

We next establish an implication of the sensitivity model that will facilitate reducing the problem to an optimization problem with a mathematically tractable form. 
\begin{lemma}
\label{lemma: model implication on latent score}
Under sensitivity model \eqref{eqn:sens_model},
\label{eq: bounded score deviation constraint}
\begin{equation}
s(a \mid x) -\gamma \leq s(a \mid x, u) \leq s(a \mid x) + \gamma, \ \forall a, x, u.
\end{equation}
\end{lemma}
\vspace{0mm}
This result demonstrates that the sensitivity models constrain the Riesz representer $s(a \mid x, u)$ to be within $\gamma$ of $s(a \mid x)$, on the additive scale. Next, we introduce a statistical restriction (entirely separate from the sensitivity model) on $s(a\mid x,u)$. The restriction is a well-known property of score functions. 
\begin{lemma}
\label{lemma: latent score integrates to marginal score}
Suppose Assumption \ref{assumption: regularity} holds, and that $f'(a \mid x, u)$ is continuous and bounded for all $a, x, u$. Then
\label{eq: score integrates constraint}
\begin{equation}
    E[s(A \mid X, U) \mid A = a, X =x] = s(a\mid x) \ \forall a, x.
\end{equation}
\end{lemma}
Thus, combining the previous two lemmas with the fact that $\theta = E[-s(A \mid X, U)Y]$, we can formulate the sensitivity analysis as an optimization problem with \eqref{eq: bounded score deviation constraint} and \eqref{eq: score integrates constraint} as constraints. As a result, a valid sensitivity analysis under model (\ref{eqn:sens_model}) for the ADE would solve the following optimization problem:
\begin{equation}
\begin{aligned}
& \underset{s(a\mid x, u)}{\text{maximize/minimize}}
& & E[-s(A \mid X, U) Y] \\
& \text{subject to} & & s(a \mid x, u) \in [s(a\mid x) - \gamma, s(a \mid x) + \gamma] \ \forall a, x, u,\\
& \text{and} & & E[s(A \mid X, U) \mid A = a, X=x] = s(a \mid x) \ \forall a, x.
\end{aligned}
\end{equation}
It is straightforward to see that the optimization can be conducted in each stratum $(A = a, X = x)$ separately. Thus, we focus on solving the following formulation:
\begin{equation}
\label{eq: optimization problem}
\begin{aligned}
& \underset{s(a\mid x, u)}{\text{maximize/minimize}}
& & E[-s(A \mid X, U) Y \mid A = a, X = x] \\
& \text{subject to} & & s(a \mid x, u) \in [s(a\mid x) - \gamma, s(a \mid x) + \gamma] \ \forall a, x, u,\\
& \text{and} & & E[s(A \mid X, U) \mid A = a, X=x] = s(a \mid x) \ \forall a, x.
\end{aligned}
\end{equation}
These optimization problems formulation follow a Lagrangian formulation, and they resemble other optimization problems in the causal inference literature, for example those in \cite{Jin2022SensitivityPerspective}, \cite{Zhang2022a}, \cite{Dorn2023}, among others. Equipped with this formulation of the optimization problem, we will aim to obtain closed-form solutions for the cases where $Y$ is continuous or binary.

\section{Solving the optimization problems}
\label{section: optimization}
We now proceed to characterizing the optimal solution and value of the optimization problem in the previous section.
\subsection{Continuous Outcome}
We first consider the case where the outcome $Y$ is continuously distributed, i.e. for all $a, x$, the distribution $Y \mid A, X$ has no point masses.
\begin{proposition}
\label{prop: solution continuous}
Suppose the outcome $Y$ is continuously distributed with no point masses. The solution to the maximization version of \eqref{eq: optimization problem} is 
\begin{equation*}
    s^*(A \mid X, U) = \begin{cases}
        s(A \mid X) - \gamma & \text{if } Y > M( A,X)\\
        s(A \mid X) + \gamma & \text{if } Y < M( A,X)
    \end{cases},
\end{equation*}
and the opposite (swap the $>$ and $<$ signs in the piecewise function) for the minimization, where $M( A,X)$ is the conditional median (1/2 quantile) of $Y$ given $A, X$. Moreover, the maximum and minimum objective values of the optimization programs are, respectively,
\begin{equation}
\label{eq: optimal values - continuous}
\begin{aligned}
\psi_{\max} = E[-s(A \mid X)Y] + \gamma E[Y (\mathbbm{1}_{\{Y > M( A, X)\}} - \mathbbm{1}_{\{Y < M( A, X)\}})], \\ \psi_{\min} = E[-s(A \mid X)Y] - \gamma E[Y (\mathbbm{1}_{\{Y > M( A, X)\}} - \mathbbm{1}_{\{Y < M( A, X)\}})].
\end{aligned}
\end{equation}
\end{proposition}
One might observe that the solution has a Neyman-Pearson flavor. This flavor of solution in sensitivity analysis has been observed before, for example for the average treatment effect in the binary treatment case \citep{Dorn2023, Zhang2022a}. In addition, the optimal values of the optimization problem equal $E[-s(A \mid X)Y]$ (what one would estimate if the unmeasured confounding due to $U$ is ignored), plus or minus $\gamma$ times a nonnegative correction term. Thus, the bounds are symmetric around $E[-s(A \mid X)Y]$. It is clear that the correction term is nonnegative.

\subsection{Binary Outcome}
The above formulation and closed form solution required the outcome $Y$ to be continuous. This can be seen from the fact that the optimal choices for $s$ depend on $Y$ being above or below some median cutoff point. In the binary case, for strata of $(A,X)$ where $0 < P(Y = 1 \mid A, X) < 1$, one cannot simply take $s^*(A \mid X, U)$ to match the form in Proposition \ref{prop: solution continuous}, replacing $M( A, X)$ with 1/2, as this will lead to a violation of the constraint on the score in Equation \eqref{eq: optimization problem} unless $P(Y = 1 \mid A, X) = 1/2$ exactly. Of course, the bound obtained by replacing $M(A,X)$ in the solution from Proposition \ref{prop: solution continuous} with 1/2 would still be valid, but potentially conservative. Instead, we can show the following result for the binary outcome case:
\begin{proposition}
\label{prop: solution binary}
Suppose the outcome $Y$ is binary. A solution to the maximization version of \eqref{eq: optimization problem} is 
\begin{equation*}
    s^*(A \mid X, U) = \begin{cases}
        s(A \mid X) - \gamma & \text{if } Y = 1\\
        s(A \mid X) + \gamma \times P(Y = 1 \mid A, X)/(1-P(Y = 1 \mid A, X)) & \text{if } Y = 0
    \end{cases},
\end{equation*}
if $P(Y = 1 \mid A, X) \leq 1/2$ and 
\begin{equation*}
    s^*(A \mid X, U) = \begin{cases}
        s(A \mid X) - \gamma \times (1-P(Y = 1 \mid A, X)) / P(Y = 1 \mid A, X) & \text{if } Y = 1\\
        s(A \mid X) + \gamma  & \text{if } Y = 0
    \end{cases},
\end{equation*} 
if $P(Y = 1 \mid A, X) > 1/2$. 
The optimal solution for the minimum is the opposite (swap the $Y = 1$ and $Y = 0$ solutions). Moreover, the maximum and minimum values of the optimization program are, respectively,
\begin{equation}
\label{eq: optimal values - binary}
\begin{aligned}
\psi_{\max}^B  = E[-s(A\mid X)Y] + \gamma E[1/2 - |P(Y = 1 \mid A,X) - 1/2|], \\ \psi_{\min}^B  = E[-s(A\mid X)Y] - \gamma E[1/2 - |P(Y = 1 \mid A,X) - 1/2|].
\end{aligned}
\end{equation}
\end{proposition}
Again, the solution has a Neyman-Pearson flavor and the optimal value takes the form of $E[-s(A \mid X)Y]$ (what one would estimate if they ignored $U$), plus or minus $\gamma$ times a nonnegative correction term, which increases as $P(Y = 1 \mid A, X)$ approaches 1/2. In the binary case, however, the correction term is not as simple to estimate. We observe that the optimal value for the binary outcome involves an absolute value (or maximum), because of the term $\gamma E[1/2 - |P(Y = 1 \mid A,X) - 1/2|] = \gamma E[1/2 - \max\{P(Y = 1 \mid A,X) - 1/2, 1/2 - P(Y = 1 \mid A,X) \}] = \gamma E[\min\{1 - P(Y = 1 \mid A,X) , P(Y = 1 \mid A,X) \}]$. Such a term is not smooth when the probability that $P(Y = 1 \mid A,X) = 1/2$ is not zero. An approach that is popular when trying to estimate such non-smooth quantities is to instead target a smooth approximation that bounds the true quantity of interest, which we describe in detail in Section \ref{section: estimation and inference}. 

\begin{remark}
A natural question to ask is whether the bounds derived in the previous two sections are ``sharp'' in some sense. For example, in the binary treatment case, for the marginal model from \cite{Tan2006AScores}, \cite{Dorn2023} exhibited sharpness of their bounds (and looseness of the original bounds of \cite{Zhao2019SensitivityBootstrap}) by constructing potential outcome and $U$ distributions that can achieve their bounds simultaneously for the treated and control potential outcomes while producing an identical observed data distribution. It is difficult to find an analogous construction in our more complicated setting with a continuous exposure and a sensitivity model that simultaneously places bounds on quantities at all pairs of dose levels $a, a'$. Moreover, the optimization problem in \eqref{eq: optimization problem} is solved separately for each stratum $(a,x)$, so it does not impose the mean zero (averaging over $A$) restriction on score functions. Nevertheless, we do impose the constraint implied by Lemma \ref{lemma: latent score integrates to marginal score}, which is in some sense analogous to the additional constraint imposed by \cite{Dorn2023} to sharpen the original conservative bounds of \cite{Zhao2019SensitivityBootstrap}. In addition, the simulations and real data applications suggest that our bounds can be informative.
\end{remark}

\section{Estimation and inference}
\label{section: estimation and inference}
For estimation and inference for the closed-form bounds, we appeal to semiparametric efficiency theory \citep{Tsiatis2006} and debiased machine learning \citep{Chernozhukov2018b}. As alluded to previously, the bounds for a binary outcome can be non-smooth, so we instead target a smooth approximation. The central object in semiparametric efficiency theory is the efficient influence function, whose variance equals the semiparametric efficiency bound, and is unique in a completely nonparametric model. The rest of this section is devoted to characterizing the efficient influence functions of the (smoothed) bounds for continuous and binary outcomes in the nonparametric model, which will motivate construction of estimators. For convenience, we will at times refer to the efficient influence function even when the precise terminology would be the \emph{uncentered} efficient influence function. For a recent review of semiparametric theory and the connection to debiased machine learning, we refer the reader to \cite{Kennedy2022} and \cite{chen2026equivalence}, respectively.
\subsection{Continuous outcome}
In this subsection, we will derive the efficient influence function for the bounds on the ADE under the sensitivity model with continuous outcomes. Recall that these were $\psi_{\max} = E[-s(A \mid X)Y] + \gamma E[Y (\mathbbm{1}_{\{Y > M( A, X)\}} - \mathbbm{1}_{\{Y < M( A, X)\}})]$ and $\psi_{\min} = E[-s(A \mid X)Y] - \gamma E[Y (\mathbbm{1}_{\{Y > M( A, X)\}} - \mathbbm{1}_{\{Y < M( A, X)\}})]$. The efficient influence function for the functional $E[-s(A \mid X)Y]$ was derived in \cite{Newey1993EfficiencyModels}. Thus, by a linearity property of efficient influence functions \citep{Kennedy2022}, it only remains to find the efficient influence function of $\gamma E[Y (\mathbbm{1}_{\{Y > M( A, X)\}} - \mathbbm{1}_{\{Y < M( A, X)\}})]$.

\begin{proposition}
\label{prop: eif correction continuous}
The efficient influence function of $\gamma E[Y (\mathbbm{1}_{\{Y > M( A, X)\}} - \mathbbm{1}_{\{Y < M( A, X)\}})]$ is 
\vspace{0mm}
\begin{align*}
\gamma ((Y - M(A,X))&\mathbbm{1}_{\{Y > M( A, X)\}} +1/2M(A,X) - (Y - M(A,X))\mathbbm{1}_{\{Y < M( A, X)\}} - 1/2M(A,X)) \\ &= \gamma (Y - M(A,X))(\mathbbm{1}_{\{Y > M( A, X)\}} - \mathbbm{1}_{\{Y < M( A, X)\}}).
\end{align*}
\end{proposition}
\vspace{0mm}
The efficient influence function of $E[-s(A \mid X)Y]$ derived in \cite{Newey1993EfficiencyModels} takes the following form:
\begin{equation*}
     \mu'(A,X) - s(A \mid X)\{Y - \mu(A,X)\},
\end{equation*}
where $\mu(A, X) \equiv E[Y \mid A, X]$, and $\mu'(A,X) \equiv \partial_a \mu(A,X)$. We then get the immediate corollary:
\begin{corollary}
\label{cor: eif bounds continuous}
The efficient influence functions of $\psi_{\max}$ and $\psi_{\min}$ under the nonparametric model are respectively,
\begin{equation}
\begin{aligned}
\phi_{\text{max}} = \mu'(A,X) - s(A \mid X)\{Y - \mu(A,X)\} + \gamma (Y - M(A,X))(\mathbbm{1}_{\{Y > M( A, X)\}} - \mathbbm{1}_{\{Y < M( A, X)\}}), \\ \phi_{\text{min}} = \mu'(A,X) - s(A \mid X)\{Y - \mu(A,X)\} - \gamma (Y - M(A,X))(\mathbbm{1}_{\{Y > M( A, X)\}} - \mathbbm{1}_{\{Y < M( A, X)\}}). 
\end{aligned}
\end{equation}
\end{corollary}
Equipped with the efficient influence functions, it is straightforward to propose estimators with desirable properties. The estimator will require estimating the unknown nuisance functions $\mu, \mu', s, M$. Compared to estimating the ADE assuming no unmeasured confounding, only the conditional median $M$ is an additional nuisance. To ease the notational burden for the theoretical analysis, we simply analyze a sample split estimator where the nuisance functions are estimated on one split of the data, and on the second split, those estimates are plugged in to the efficient influence function at each data point, i.e.
\begin{equation}
\label{eq: continuous estimator}
\begin{aligned}
\wh{\psi}_{\text{max}} &= \frac{1}{n} \sum_{i = 1}^n \wh{\mu}'(A_i,X_i) - \wh{s}(A_i \mid X_i)\{Y_i - \wh{\mu}(A_i,X_i)\} \\ &+ \gamma (Y_i - \wh{M}(A_i,X_i))(\mathbbm{1}_{\{Y_i > \wh{M}( A_i, X_i)\}} - \mathbbm{1}_{\{Y_i < \wh{M}( A_i,X_i)\}}),\\
\wh{\psi}_{\text{min}}  &= \frac{1}{n} \sum_{i = 1}^n \wh{\mu}'(A_i,X_i) - \wh{s}(A_i\mid X_i)\{Y_i - \wh{\mu}(A_i,X_i)\} \\ &- \gamma (Y_i - \wh{M}(A_i,X_i))(\mathbbm{1}_{\{Y_i > \wh{M}( A_i, X_i)\}} - \mathbbm{1}_{\{Y_i < \wh{M}( A_i,X_i)\}}).
\end{aligned}
\end{equation}
To make use of all of the data, one can employ the now commonly utilized cross-fitting technique \citep{Chernozhukov2018b}, where the data is randomly split into $K$ roughly equally sized folds $D_1, \ldots, D_K$. For each $k = 1,\ldots, K$, we compute nuisance estimates $\wh{\eta} = (\wh{\mu}, \wh{\mu'}, \wh{s}, \wh{M})$ for $\eta = (\mu, \mu', s, M)$ on all folds except $D_k$, and plug in these estimates on fold $D_k$ (as in \eqref{eq: continuous estimator}). The result of Theorem \ref{thm: asymptotic normality continuous}, which outlines conditions under which the sample split estimator achieves asymptotic normality, will also apply to an analogous cross-fitted estimator.
\begin{theorem}
\label{thm: asymptotic normality continuous}
Suppose nuisance estimates $\wh{\eta}$ are estimated from an independent sample of data and the conditional densities $f(y \mid a, x)$ are uniformly bounded with no point masses. Also, assume that $(\wh{\mu}, \wh{\mu'}, \wh{s}, \wh{M}, \mu, \mu', s, M)$ are bounded almost surely and that $\norm{\wh{\mu} - \mu} + \norm{\wh{\mu'} - \mu'} + \norm{\wh{s} - s} + \norm{\wh{M} - M} = o_p(n^{-1/4})$. Then  
\begin{align*}
\sqrt{n}(\wh{\psi}_{\textmax} - \psi_{\textmax}) \stackrel{d}{\to} N(0, \text{Var}(\phi_{\textmax})), \\ 
\sqrt{n}(\wh{\psi}_{\text{min}} - \psi_{\text{min}}) \stackrel{d}{\to} N(0, \text{Var}(\phi_{\text{min}})).\\ 
\end{align*}
\end{theorem}
\vspace{0mm}
By virtue of using an estimator based on the efficient influence function, the bias of the estimator only involves second-order nuisance estimation errors. Thus, $\sqrt{n}$ consistency is possible as long as the nuisance functions can be estimated at the (slower than parametric) $n^{-1/4}$ rate. This makes it possible to conduct valid inference even when using nonparametric or data-adaptive machine learning estimates of the nuisance functions, provided they are not converging too slowly to the truth. Based on the asymptotic normality of the estimators, it is straightforward to construct Wald-style confidence intervals. Since we are bounding upper and lower bounds, it is reasonable to construct one-sided confidence intervals. Explicitly, 
\begin{equation}
\label{eq: continuous ci}
\begin{aligned}
\wh{\psi}_{\textmax} + z_{1 - \alpha} \sqrt{\wh{\sigma}_{\textmax} / n}, \\
\wh{\psi}_{\textmin} - z_{1 - \alpha} \sqrt{\wh{\sigma}_{\textmin} / n},
\end{aligned}
\end{equation}
are asymptotically valid $1 - \alpha$ confidence upper and lower bounds for $\psi_{\textmax}$ and $\psi_{\textmin}$, provided the variance estimates $\wh{\sigma}_{\textmax}$ and $\wh{\sigma}_{\textmin}$ converges to the true variances. Here, plug-in variance estimates can be used:
\begin{equation*}
\scriptsize
\begin{aligned}
&\wh{\sigma}_{\textmax} \equiv \frac{1}{n} \sum_{i = 1}^n \left\{ \wh{\mu}'(A_i,X_i) - \wh{s}(A_i\mid X_i)\{Y_i - \wh{\mu}(A_i,X_i)\} + \gamma (Y_i - \wh{M}(A_i,X_i))(\mathbbm{1}_{\{Y_i > \wh{M}( A_i, X_i)\}} - \mathbbm{1}_{\{Y_i < \wh{M}( A_i,X_i)\}}) - \wh{\psi}_{\textmax}\right\}^2, \\ 
&\wh{\sigma}_{\textmin} \equiv \frac{1}{n} \sum_{i = 1}^n \left\{ \wh{\mu}'(A_i,X_i) - \wh{s}(A_i\mid X_i)\{Y_i - \wh{\mu}(A_i,X_i)\} - \gamma (Y_i - \wh{M}(A_i,X_i))(\mathbbm{1}_{\{Y_i > \wh{M}( A_i, X_i)\}} - \mathbbm{1}_{\{Y_i < \wh{M}( A_i,X_i)\}}) - \wh{\psi}_{\textmin}\right\}^2.
\end{aligned}
\end{equation*}
These plug-in estimates are consistent under the assumptions of Theorem \ref{thm: asymptotic normality continuous}, and so the confidence bounds from Equation \eqref{eq: continuous ci} will be asymptotically valid.
\subsection{Binary outcome}
As derived in the previous section, the bounds in the binary outcome case involve the term $\gamma E[\min\{1 - P(Y = 1 \mid A,X) , P(Y = 1 \mid A,X) \}]$. The presence of the minimum makes this quantity potentially non-smooth, so we instead rely on a smooth approximation. Specifically, when minima or maxima are involved, the LogSumExp (LSE) function is a popular choice; see \cite{Levis2023} for a recent example in causal inference. For a minimum of $k$ quantities, and any fixed $t > 0$,
\begin{equation}
\label{eq: LSE approximation}
\min\{x_1,\ldots,x_k\} - \frac{\log(k)}{t} \leq -\frac{1}{t} \log(\exp(-tx_1)+\ldots+\exp(-tx_k)) \leq \min\{x_1,\ldots,x_k\}.
\end{equation}
Specialized to our setting, where we take a minimum of $p$ and $1-p$, we define 
\begin{equation}
h_t(p) \equiv  -\frac{1}{t} \log(\exp(-tp)+\exp(-t(1-p))).
\end{equation}
Thus, we will instead estimate (for a fixed $t$)
\begin{equation}
\label{eq: optimal values - binary relaxed}
\begin{aligned}
\psi_{\max,h_t}^{B} \equiv E[-s(A\mid X)Y] + \gamma E[h_t(P(Y = 1 \mid A,X))], \\ 
\psi_{\min,h_t}^{B} \equiv E[-s(A\mid X)Y] -\gamma E[h_t(P(Y = 1 \mid A,X))].
\end{aligned}
\end{equation}
From Equation \eqref{eq: LSE approximation}, it is immediate that $\psi_{\max,h_t}^{B} \leq \psi_{\max}^B \leq \psi_{\max,h_t}^{B} + \gamma \log(2)/t$ and $\psi_{\min,h_t}^{B} \geq \psi_{\min}^B \geq \psi_{\min,h_t}^{B} - \gamma \log(2)/t$ for any $t > 0$, and the inequality gap shrinks as $t$ increases. At the same time, $h_t$ becomes less smooth as $t$ increases, and consequently, $E[h_t(P(Y = 1 \mid A,X))]$ becomes harder to estimate. Thus, in choosing $t$, there is a trade-off between approximation error and statistical estimation. A rigorously justified ``optimal'' choice for $t$ is outside the scope of this paper, but we refer the reader to \cite{Levis2023} for some additional discussion.

In the remainder of this section, we will derive the efficient influence function for the smoothed lower and upper bounds for the ADE under the sensitivity model with binary outcomes. Recall that these were $\psi_{\textmax,h_t}^{B}$ and $\psi_{\textmin,h_t}^{B}$. As in the continuous outcome case, the efficient influence function for the functional $E[-s(A \mid X)Y]$ was derived in \cite{Newey1993EfficiencyModels}. Thus, it only remains to find the efficient influence function of $E[h_t\left(P(Y = 1 \mid A,X)\right)]$.

\begin{proposition}
\label{prop: eif correction binary}
The efficient influence function of $\theta_h^B \equiv E[h\left(P(Y = 1 \mid A,X)\right)]$ for any univariate, continuously differentiable function $h$ is given by
\begin{equation*}
h\left(P(Y = 1 \mid A,X)\right) + h'(P(Y = 1 \mid A,X))(Y - P(Y = 1 \mid A,X)).
\end{equation*}
\end{proposition}
\vspace{0mm}
As before, we get the immediate corollary:
\begin{corollary}
\label{cor: eif bounds binary}
The efficient influence functions for $\psi_{\textmax,h_t}^{B}$ and $\psi_{\textmin,h_t}^{B}$ are, respectively,
\begin{align*}
\phi_{\textmax,h_t}^B &= \mu'(A,X) - s(A \mid X)\{Y - \mu(A,X)\} + \\ &\gamma \{h_t\left(P(Y = 1 \mid A,X)\right) + h_t'(P(Y = 1 \mid A,X))(Y - P(Y = 1 \mid A,X))\},
\\
\phi_{\textmin,h_t}^B &= \mu'(A,X) - s(A \mid X)\{Y - \mu(A,X)\} - \\ &\gamma \{h_t\left(P(Y = 1 \mid A,X)\right) + h_t'(P(Y = 1 \mid A,X))(Y - P(Y = 1 \mid A,X))\}.
\end{align*}
\end{corollary}
\vspace{0mm}
Again, we propose estimators based on the efficient influence function, and present sample-split versions that estimate the nuisances $\mu, \mu', s$ (these are the same nuisances needed to estimate the ADE assuming no unmeasured confounding). They are as follows (since $Y$ is binary, $\wh{\mu}(A_i,X_i)$ and $\wh{P}(Y = 1 \mid A_i,X_i)$ are equivalent):
\begin{equation}
\begin{aligned}
\wh{\psi}_{\textmax,h_t}^B &= \frac{1}{n} \sum_{i = 1}^n \wh{\mu}'(A_i,X_i) - \wh{s}(A_i \mid X_i)\{Y_i - \wh{\mu}(A_i,X_i)\} \\ &+ \gamma (h_t\left(\wh{P}(Y = 1 \mid A_i,X_i)\right) + h_t'(\wh{P}(Y = 1 \mid A_i,X_i))(Y_i - \wh{P}(Y = 1 \mid A_i,X_i))),\\
\wh{\psi}_{\textmin,h_t}^B  &= \frac{1}{n} \sum_{i = 1}^n \wh{\mu}'(A_i,X_i) - \wh{s}(A_i \mid X_i)\{Y_i - \wh{\mu}(A_i,X_i)\} \\ &- \gamma (h_t\left(\wh{P}(Y = 1 \mid A_i,X_i)\right) + h_t'(\wh{P}(Y = 1 \mid A_i,X_i))(Y_i - \wh{P}(Y = 1 \mid A_i,X_i))).
\end{aligned}
\end{equation}
Theorem \ref{thm: asymptotic normality binary} establishes asymptotic normality of the estimators.
\begin{theorem}
\label{thm: asymptotic normality binary}
Suppose nuisance estimates $\wh{\eta}$ are estimated from an independent sample of data. Also, assume that $(\wh{\mu}, \wh{\mu'}, \wh{s},  \mu', s)$ are bounded almost surely and that $\norm{\wh{\mu} - \mu} + \norm{\wh{\mu'} - \mu'} + \norm{\wh{s} - s} = o_p(n^{-1/4})$. Then  
\begin{align*}
&\sqrt{n}(\wh{\psi}_{\textmax,h_t}^B - \psi_{\textmax, h_t}^B) \stackrel{d}{\to} N(0, \text{Var}(\phi_{\textmax,h_t}^B)), \text{and} \\ 
&\sqrt{n}(\wh{\psi}_{\textmin, h_t}^B - \psi_{\textmin,h_t}^B) \stackrel{d}{\to} N(0, \text{Var}(\phi_{\textmin,h_t}^B)).\\ 
\end{align*}
\end{theorem}
\vspace{0mm}
Similar to the continuous outcome case, we can construct asymptotically valid Wald-style confidence intervals for $\psi_{\max}^B$ and $\psi_{\min}^B$ (rather than the smoothed $\psi_{\textmax, h_t}^B$ and $\psi_{\textmin, h_t}^B$), if we account for the approximation error from the LSE function. Moreover, as in the continuous outcome case, $\sqrt{n}$ consistency is possible as long as the nuisance functions can be estimated at the  $n^{-1/4}$ rate. The respective upper and lower bounds for the $1 - \alpha$ confidence intervals are 
\begin{equation}
\label{eq: binary ci}
\begin{aligned}
\wh{\psi}_{\textmax,h_t}^B + z_{1 - \alpha} \sqrt{\wh{\sigma}_{\textmax, h_t}^B / n} + \gamma \log(2) / t, \\
\wh{\psi}_{\textmin,h_t}^B - z_{1 - \alpha} \sqrt{\wh{\sigma}_{\textmin, h_t}^B / n} - \gamma \log(2) / t.
\end{aligned}
\end{equation}
\vspace{0mm}
Again, plug-in variance estimates can be used:
\begin{equation*}
\scriptsize
\begin{aligned}
&\wh{\sigma}_{\textmax, h_t}^B \equiv \frac{1}{n} \sum_{i = 1}^n \left\{ \wh{\mu}'(A_i,X_i) - \wh{s}(A_i \mid X_i)\{Y_i - \wh{\mu}(A_i,X_i)\} + \gamma (h_t\left(\wh{\mu}(A_i,X_i)\right) + h_t'(\wh{\mu}(A_i,X_i))(Y_i - \wh{\mu}(A_i,X_i))) - \wh{\psi}_{\textmax,h_t}^B\right\}^2, \\ 
&\wh{\sigma}_{\textmin, h_t}^B \equiv \frac{1}{n} \sum_{i = 1}^n \left\{ \wh{\mu}'(A_i,X_i) - \wh{s}(A_i \mid X_i)\{Y_i - \wh{\mu}(A_i,X_i)\} - \gamma (h_t\left(\wh{\mu}(A_i,X_i)\right) + h_t'(\wh{\mu}(A_i,X_i))(Y_i - \wh{\mu}(A_i,X_i))) - \wh{\psi}_{\textmin,h_t}^B\right\}^2.
\end{aligned}
\end{equation*}
These estimates will be consistent under the assumptions of Theorem \ref{thm: asymptotic normality binary}, and so the confidence bounds from Equation \eqref{eq: binary ci} will be asymptotically valid.

\subsection{Simultaneous confidence bands}
The previous subsections introduced estimators and pointwise confidence intervals for a fixed value of $\gamma$. In this subsection, we briefly outline how to conduct simultaneous inference when we wish to conduct the sensitivity analysis over a bounded interval of values, i.e. $\gamma \in [\gamma_l, \gamma_u]$. Conveniently, the nature of the resulting estimands take the form $a \pm \gamma b$, and we can easily construct Wald confidence intervals for $a$ and $b$ under the same assumptions as in the previous subsections. Thus, by the union bound, a uniform $100(1-\alpha)$\% confidence band can be straightforwardly constructed for $\gamma \in [\gamma_l, \gamma_u]$ by constructing $100(1-\alpha/2)$\% confidence intervals for $a$ (two-sided) and $b$ (one-sided upper limit) and concatenating accordingly. In our setting, $a$ corresponds to $E[-s(A \mid X)Y]$ and $b$ corresponds to either $E[Y (\mathbbm{1}_{\{Y > M( A, X)\}} - \mathbbm{1}_{\{Y < M( A, X)\}})]$ (continuous $Y$) or $E[h_t(P(Y = 1 \mid A,X))] + \log(2)/t$ (binary $Y$). Wald-style $100(1-\alpha/2)$\% confidence intervals for $a$ and $b$ can be constructed using the respective efficient influence functions. Alternatively, a multiplier bootstrap approach could be implemented; see \cite{Kennedy2019a} and \cite{Zhang2022a} for recent applications in causal inference. However, we do not pursue that direction as the proposed approach is much simpler.

\section{Simulations}
\label{section: sims}
We now evaluate the finite sample performance of the proposed methods through a simulation study. In the simulation, we assess the coverage of confidence intervals of the true ADE when there \emph{is} unmeasured confounding based on sensitivity analyses at different choices of $\gamma$.
\subsection{Simulation setup}
We conduct simulations corresponding to two different dose distributions, two outcome types (binary and continuous), and three different strengths of unmeasured confounders, yielding $ 2 \times 2 \times 3 = 12$ different settings. For all 12 settings, we draw the confounders $X \sim \text{Unif}[0,1]^d$, $d = 5$, and draw $U \mid X \sim \text{Bern}(\Phi(\sin(X_1+X_2)))$. We draw the dose from a conditional density that is Gaussian or Gamma. For the Gaussian case, $A \mid X, U \sim N(\alpha^TX+\zeta U, 1)$. For the Gamma case (shape and rate parametrization), $A \mid X, U \sim \text{Gamma}(13, 8 + \alpha^T X - \zeta U)$. $\zeta$ is set to $\log(2)$ for all simulations, and here $\alpha$ are randomly drawn coefficients from $N(0,1)$, redrawn at each iteration. For the outcome model, we consider different settings for a binary outcome and for a continuous outcome. For the continuous outcome, we draw $Y \mid A, X, U \sim N(\kappa A + \beta^T X + \tau U  + \beta^T_{AX}AX, 1)$. For the binary outcome, we use a probit model and draw $Y \mid A, X, U \sim \text{Bern}(\Phi(\kappa A + \beta^T X + \tau U + \beta^T_{AX}AX))$. $\tau$ is varied in $\{2, 3, 4\}$ and $\kappa = 0.25$. For both outcome models, the $\beta$ coefficients are randomly drawn at each iteration from $N(-1,1)$. The interaction coefficients $\beta^T_{AX}$ are randomly drawn from $N(0,\text{sd} = 1/4)$, redrawn at each iteration. In the simulation, a higher $U$ leads to a higher chance of both a higher dose and outcome. The derivatives of the conditional expectations $E[Y \mid A, X, U]$ are available in closed form, and so the ``ground truth'' average derivative effects are approximated by drawing $10^7$ Monte-Carlo samples from the joint distribution of $(A, X, U)$ and computing the sample average of the derivative of $E[Y \mid A, X, U]$. 
\subsection{Estimators}
The nuisance estimates required to compute the estimator include the score function $s(a \mid x)$, the conditional mean $\mu(a,x)$ and its derivative $\mu'(a,x)$, and the conditional median $M(a,x)$ for the continuous outcome case. For the binary outcome case, we set $t = 50$ for computing the LSE function approximation. We use the \texttt{R} packages \texttt{drape} and \texttt{xgboost} for nuisance estimation. Specifically, we estimate scores $s(a \mid x)$ and conditional mean derivatives $\mu'(a,x)$ using adaptations of the methods from the \texttt{drape} package \citep{Klyne2023AverageLearning}. The methods proposed by \cite{Klyne2023AverageLearning} can re-smooth any first-stage regression $\wh{\mu}(a,x)$ estimator to produce a differentiable version to obtain an estimate $\wh{\mu}'(a,x)$, and model the conditional distribution $f(a \mid x)$ through a location-scale model to estimate $s(a \mid x)$. Hyperparameters for these methods were chosen in the same manner as the simulations in \cite{Klyne2023AverageLearning}. To fit conditional means $\mu(a,x)$ and medians $M(a,x)$, we use gradient boosted trees as implemented in the \texttt{xgboost} package with default hyperparameters and the appropriate loss function -- squared error for estimating the conditional mean of a continuous variable, logistic loss for estimating the conditional mean of a binary variable, and absolute error for estimating the conditional median of a continuous variable. To make use of the full data sample, we implement 5-fold cross-fitting.

\subsection{Simulation results}
Table \ref{tab: simulation results} collects coverage results for pointwise 95\% confidence intervals of the sensitivity analysis procedures at varying levels of $\gamma$. One can verify that the $\gamma$ at which the sensitivity analysis model \eqref{eqn:sens_model} holds (and thus the procedure will be valid) is between $0.5\log(2)$ and $\log(2)$. Therefore, it is not surprising to see in Table \ref{tab: simulation results} that the 95\% sensitivity analysis confidence intervals can severely undercover when $\gamma$ is taken to be $0$ (no unmeasured confounding) or $0.25 \log(2)$, as both of these are less than the lower bound $0.5\log(2)$. This also gives some reassurance that although we have not established sharpness of the analytic bounds, the bounds can still be informative. In addition, one may notice that as the strength of the unmeasured confounders impact on the outcome, measured through $\tau$, increases, the sensitivity analysis intervals cover less. This is expected, as the sensitivity model we consider only restricts $U$'s impact on the treatment. Thus, the sensitivity analysis must protect against arbitrary dependence between $U$ and the (potential) outcomes, i.e. arbitrarily large values of $\tau$. Consequently, it is reasonable to expect that if $\tau$ were to be increased further, the coverage rate of the sensitivity analysis bounds for $\gamma \geq 0.5 \log(2)$ would move closer towards but not necessarily reach the nominal level.
\begin{table}[ht]
\centering
\begin{tabular}{|llc|ccccc|}
  \hline \multicolumn{3}{|c|}{Simulation parameters} & \multicolumn{5}{|c|}{$\gamma$ value of the sensitivity analysis} \\
dose & outcome & $\tau$ &  $0$ & $0.25 \log(2)$ & $0.5 \log(2)$ & $0.75 \log(2)$ & $\log(2)$ \\ 
  \hline
Gaussian & binary & 2.00 & 0.68 & 0.95 & 0.99 & 1.00 & 1.00 \\ 
Gaussian & binary & 3.00 & 0.51 & 0.89 & 0.99 & 1.00 & 1.00 \\ 
Gaussian & binary & 4.00 & 0.26 & 0.76 & 0.99 & 1.00 & 1.00 \\ 
Gamma & binary & 2.00 & 0.88 & 0.95 & 0.98 & 0.99 & 0.99 \\ 
Gamma & binary & 3.00 & 0.80 & 0.93 & 0.98 & 0.99 & 1.00 \\ 
Gamma & binary & 4.00 & 0.70 & 0.88 & 0.97 & 0.99 & 1.00 \\ 
Gaussian & continuous & 2.00 & 0.24 & 0.95 & 1.00 & 1.00 & 1.00 \\ 
Gaussian & continuous & 3.00 & 0.07 & 0.86 & 1.00 & 1.00 & 1.00 \\ 
Gaussian & continuous & 4.00 & 0.02 & 0.73 & 1.00 & 1.00 & 1.00 \\ 
Gamma & continuous & 2.00 & 0.75 & 0.97 & 1.00 & 1.00 & 1.00 \\ 
Gamma & continuous & 3.00 & 0.72 & 0.98 & 1.00 & 1.00 & 1.00 \\ 
Gamma & continuous & 4.00 & 0.69 & 0.96 & 1.00 & 1.00 & 1.00 \\ 
\hline
\end{tabular}
\caption{95\% confidence interval coverage results for the sensitivity analysis under different sets of simulation parameters and different choices of $\gamma$. The proportions are averaged over 500 iterations.}
\label{tab: simulation results}
\end{table}

\section{Applications}
\label{section: application}
\subsection{The effect of parental income on child's education}
We illustrate the methodology for binary outcomes using an empirical example studying the extent to which household income affects a child's educational attainment \citep{Lundberg2023}. The data we use comes from the National Longitudinal Survey of Youth 1997 cohort (NLSY97).
The NLSY97 is a dataset consisting of a probability sample of U.S. youths ages 12–17, starting in 1997, who were followed up through 2019. We largely follow \cite{Lundberg2023} in pre-processing the data. The treatment variable of interest is reported total gross household income in 1996, when the respondents were age 12-17. \cite{Lundberg2023} logged and adjusted these measures to 2022 dollars using the
Consumer Price Index. Those without income measurements are dropped, as are households coded as the maximum and minimum income values, as these represent upper and lower cutoffs, not actual incomes. The outcome of interest is a report of enrollment in any college up to age 21, which is binary. Those that did not complete a survey at ages 19–21 are omitted. Following \cite{Lundberg2023}, four measured confounding variables are adjusted for: race, gender, parents’ education, and wealth. The racial categories from 1997 were Hispanic, Non-Hispanic Black, and Non-Hispanic white or other. Parents’ education is categorical, with the 3 categories no parent completed college, one parent completed college, or two parents completed college. Wealth is the log of household net worth reported by the parent in 1997, also adjusted to 2022 dollars. There are 5219 individuals in the final, processed dataset. Unfortunately, there may be confounders beyond race, gender, parents' education, and wealth that are not measured that affect both household income and propensity to pursue higher education. These might include things like innate ability or geographic location, both of which could be strongly related to household income and propensity to attend college. Thus, we implement our sensitivity analysis to assess the impact of hypothetical unmeasured confounders on the statistical conclusions. We use the same nuisance estimators as in the simulations for a binary outcome. The results of the sensitivity analysis are reported in Figure \ref{fig: income education sensitivity plot}. Each black dot represents a point estimate of an upper or lower bound at some value of $\gamma$. The shaded regions depict 95\% confidence intervals. Assuming no unmeasured confounding, the point estimate for the ADE is $0.108$, and the 95\% confidence interval $[0.075, 0.141]$. Recall that the outcome is binary, and the treatment is the log of household income. This means that on average, an increase of income of $\delta$ on the log scale might be expected to increase the propensity of attending any college by $\delta \times 10.8$ percentage points, assuming no unmeasured confounding. As $\gamma$ increases, lower point estimates and confidence bounds decrease. The point estimate ultimately crosses 0 at $\gamma = 0.323$, the 95\% pointwise confidence interval crosses 0 at $\gamma = 0.222$, and the 95\% uniform confidence interval crosses 0 at $\gamma = 0.197$.
\begin{figure}
    \centering
    \includegraphics[width=0.5\linewidth]{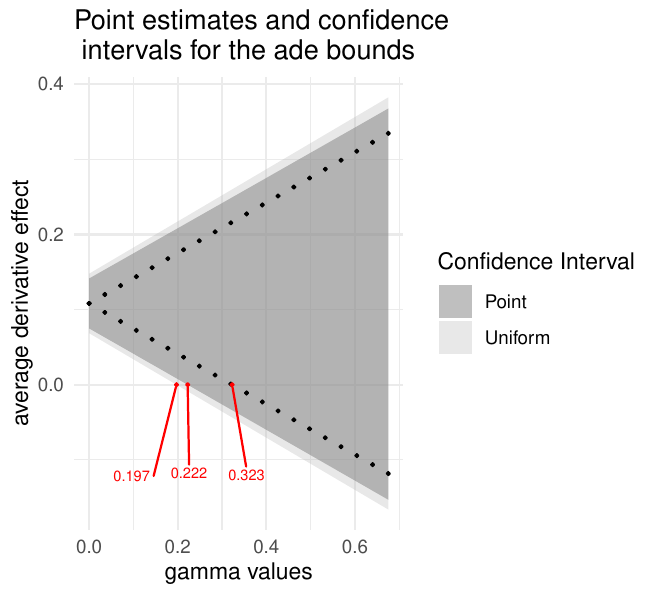}
    \caption{A plot displaying point estimates and confidence intervals for the ADE of income on probability of enrolling in college. The black dots represent point estimates of the upper and lower estimates. The darker shade represents 95\% pointwise confidence intervals at each $\gamma$ value, and the lighter shade represents 95\% simultaneous confidence intervals. The lower simultaneous confidence interval, pointwise confidence interval, and point estimate cross 0 at $\gamma = 0.197, 0.222, 0.323$ respectively.}
    \label{fig: income education sensitivity plot}
\end{figure}

\subsection{The effect of price on petrol consumption}
We now illustrate the methodology for continuous outcomes using an empirical example studying the extent to which petrol prices affect the demand for petrol, which was also studied in \cite{Chernozhukov2022a}. The data come from the Canadian National Private Vehicle Use Survey. We preprocessed the data in an identical fashion to \cite{Chernozhukov2022a}, leaving $n = 5001$ households, each of which has an outcome (log of the petrol consumption), covariates (log age, log income, log distance, and other time, geographical, household indicators), and treatment (log of petrol price per liter). In this example, the ADE measures the average price elasticity of petrol demand. 

Instead of constructing pointwise confidence intervals, we concatenated the point estimates and standard errors for $E[-s(A \mid X)Y]$ as estimated in \cite{Chernozhukov2022a} with estimates and standard errors for the correction term from Equation \eqref{eq: optimal values - continuous} to produce simultaneous confidence bands. This exercise demonstrates the ease in which the sensitivity analysis can be conducted after (and completely separate from) a primary analysis assuming no unmeasured confounding has been completed. For estimation of the correction term, the only nuisance function is the conditional median. As in the simulation, we used the \texttt{xgboost} package with absolute error loss to fit the conditional median $M(a, x)$, in conjunction with 5-fold cross-fitting. The results are displayed in Figure \ref{fig: three_side_by_side} when the generalized Dantzig selector (GDS) and Lasso estimators are used for estimating $E[-s(A \mid X)Y]$ as described in \cite{Chernozhukov2022a}. Each black dot represents a point estimate of an upper or lower bound at some value of $\gamma$. The shaded region depicts the 95\% simultaneous confidence intervals. Assuming no unmeasured confounding, the point estimates for the ADE are $-0.28$ and $-0.16$ for GDS and the Lasso, respectively. As $\gamma$ increases, upper point estimates and confidence bounds increase. The point estimates cross 0 at $\gamma = 0.924$ and $\gamma = 0.528$ and the 95\% simultaneous confidence bounds cross 0 at $\gamma = 0.524$ and $\gamma = 0.2$ for GDS and Lasso, respectively. 

\begin{figure}[h]
    \centering
    \begin{subfigure}{0.48\textwidth}
        \includegraphics[width=\linewidth]{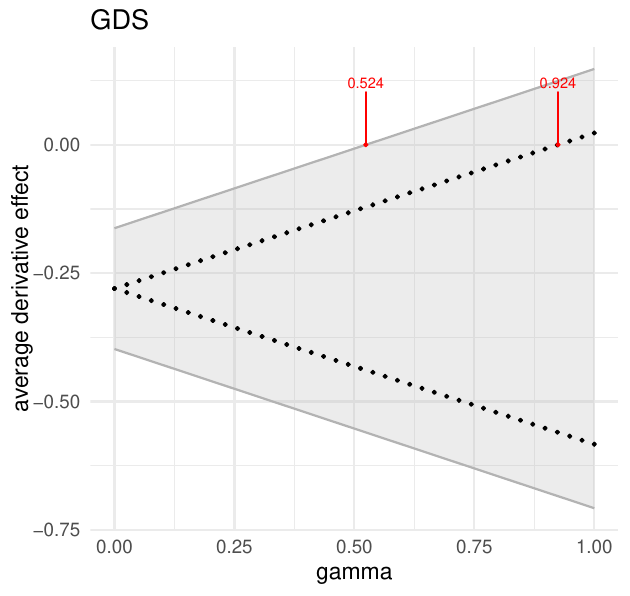}
    \end{subfigure}
    \begin{subfigure}{0.48\textwidth}
        \includegraphics[width=\linewidth]{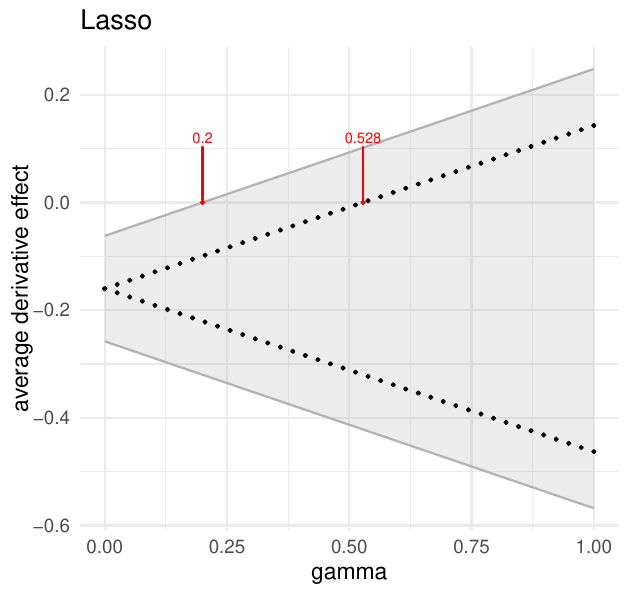}
    \end{subfigure}
   
    \caption{A plot displaying point estimates for upper and lower bounds (black dots) and 95\% simultaneous confidence intervals (shaded region) for the ADE of price on petrol consumption when using the GDS and Lasso estimators of \cite{Chernozhukov2022a} to estimate $E[-s(A \mid X)Y]$. With GDS, the upper simultaneous confidence interval and point estimate cross 0 at $\gamma = 0.524, 0.924$. With Lasso, the upper simultaneous confidence interval and point estimate cross 0 at $\gamma = 0.2, 0.528$.}
    \label{fig: three_side_by_side}
\end{figure}
\section{Discussion}
\label{section: discussion}
In this paper, we have proposed a new sensitivity model for the ADE estimand, along with valid closed-form bounds, and estimators and confidence intervals for said bounds. The form of the bounds differ for continuous and binary outcomes, and are particularly convenient, allowing easy construction of uniform confidence intervals. The extent to which the bounds introduced in this paper are sharp, in the sense of \cite{Dorn2023}, is unclear and is a promising direction for future research. To the best of our knowledge, sharp bounds for causal estimands under sensitivity models resembling \eqref{eqn:sens_model} do not exist beyond the binary treatment case. Another promising direction for further inquiry might be a calibration procedure, in the vein of \cite{Hsu2013} and \cite{McClean2024a}. Calibrating the sensitivity analysis to observed confounders, for example, could potentially help researchers gauge whether a certain magnitude of $\gamma$ is plausible, though such a practice has limitations. Finally, it may be of interest to study ordinal rather than continuous exposures. Such exposures may arise, for example, when doses of drugs are prescribed at a finite number of ordered levels. 

\section*{Acknowledgements}
We thank Abhinandan Dalal, Zhihan Huang, Ziang Niu, Zhimei Ren, Dylan Small, Eric Tchetgen Tchetgen, and participants at ACIC 2025 for helpful discussions and comments. Claude was used for proofreading purposes.

\bibliographystyle{plainnat}

\bibliography{bibliography}

\begin{thebibliography}{56}
\providecommand{\natexlab}[1]{#1}
\providecommand{\url}[1]{\texttt{#1}}
\expandafter\ifx\csname urlstyle\endcsname\relax
  \providecommand{\doi}[1]{doi: #1}\else
  \providecommand{\doi}{doi: \begingroup \urlstyle{rm}\Url}\fi

\bibitem[Baitairian et~al.(2024)Baitairian, Sebastien, Jreich, Katsahian, and Guilloux]{Baitairian2024}
Jean-Baptiste Baitairian, Bernard Sebastien, Rana Jreich, Sandrine Katsahian, and Agathe Guilloux.
\newblock {Sharp Bounds for Continuous-Valued Treatment Effects with Unobserved Confounders}.
\newblock \emph{arXiv preprint arXiv:2411.02231}, 2024.

\bibitem[Basit et~al.(2023)Basit, Latif, and Wahed]{Basit2023SensitivityTreatments}
Md~Abdul Basit, Mahbub~AHM Latif, and Abdus~S Wahed.
\newblock {Sensitivity Analysis of Inverse Probability Weighting Estimators of Causal Effects in Observational Studies with Multivalued Treatments}.
\newblock \emph{arXiv preprint arXiv:2308.15986}, 2023.

\bibitem[Bonvini and Kennedy(2022)]{Bonvini2022}
Matteo Bonvini and Edward~H. Kennedy.
\newblock {Sensitivity Analysis via the Proportion of Unmeasured Confounding}.
\newblock \emph{Journal of the American Statistical Association}, 117\penalty0 (539):\penalty0 1540--1550, 2022.

\bibitem[Bonvini et~al.(2022)Bonvini, Kennedy, Ventura, and Wasserman]{Bonvini2022SensitivityModels}
Matteo Bonvini, Edward Kennedy, Valerie Ventura, and Larry Wasserman.
\newblock {Sensitivity analysis for marginal structural models}.
\newblock \emph{arXiv preprint arXiv:2210.04681}, 2022.

\bibitem[Bruns-Smith(2025)]{bruns2025two}
David Bruns-Smith.
\newblock {Two-Stage Machine Learning for Nonparametric Instrumental Variable Regression}.
\newblock \emph{Kilts Center at Chicago Booth Marketing Data Center Paper Forthcoming}, 2025.

\bibitem[Bruns-Smith and Zhou(2023)]{bruns2023robust}
David Bruns-Smith and Angela Zhou.
\newblock {Robust fitted-q-evaluation and iteration under sequentially exogenous unobserved confounders}.
\newblock \emph{arXiv preprint arXiv:2302.00662}, 2023.

\bibitem[Chen et~al.(2026)Chen, Kennedy, and Balakrishnan]{chen2026equivalence}
Yuxi Chen, Edward~H Kennedy, and Sivaraman Balakrishnan.
\newblock On the equivalence between neyman orthogonality and pathwise differentiability.
\newblock \emph{arXiv preprint arXiv:2603.15817}, 2026.

\bibitem[Chernozhukov et~al.(2018)Chernozhukov, Chetverikov, Demirer, Duflo, Hansen, Newey, and Robins]{Chernozhukov2018b}
Victor Chernozhukov, Denis Chetverikov, Mert Demirer, Esther Duflo, Christian Hansen, Whitney Newey, and James Robins.
\newblock {Double/debiased machine learning for treatment and structural parameters}.
\newblock \emph{The Econometrics Journal}, 21\penalty0 (1):\penalty0 C1--C68, 2018.

\bibitem[Chernozhukov et~al.(2022)Chernozhukov, Newey, and Singh]{Chernozhukov2022a}
Victor Chernozhukov, Whitney~K Newey, and Rahul Singh.
\newblock {Debiased machine learning of global and local parameters using regularized Riesz representers}.
\newblock \emph{The Econometrics Journal}, 25\penalty0 (3):\penalty0 576--601, 2022.

\bibitem[Chernozhukov et~al.(2026)Chernozhukov, Cinelli, Newey, Sharma, and Syrgkanis]{Chernozhukov2022LongLearning}
Victor Chernozhukov, Carlos Cinelli, Whitney~K Newey, Amit Sharma, and Vasilis Syrgkanis.
\newblock Long story short: Omitted variable bias in causal machine learning.
\newblock \emph{The Review of Economics and Statistics}, pages 1--45, 2026.

\bibitem[Cornfield et~al.(1959)Cornfield, Haenszel, Hammond, Lilienfeld, Shimkin, and Wynder]{Cornfield2009}
Jerome Cornfield, William Haenszel, E.~Cuyler Hammond, Abraham~M. Lilienfeld, Michael~B. Shimkin, and Ernst~L. Wynder.
\newblock {Smoking and lung cancer: recent evidence and a discussion of some questions*}.
\newblock \emph{International Journal of Epidemiology}, 38\penalty0 (5):\penalty0 1175--1191, 1959.

\bibitem[Dalal and Tchetgen~Tchetgen(2025)]{dalal2025partial}
Abhinandan Dalal and Eric~J Tchetgen~Tchetgen.
\newblock {Partial Identification of Causal Effects for Endogenous Continuous Treatments}.
\newblock \emph{arXiv preprint arXiv:2508.13946}, 2025.

\bibitem[D{\'{i}}az and van~der Laan(2013)]{Diaz2013}
Iván D{\'{i}}az and Mark~J. van~der Laan.
\newblock {Sensitivity Analysis for Causal Inference under Unmeasured Confounding and Measurement Error Problems}.
\newblock \emph{The International Journal of Biostatistics}, 9\penalty0 (2), 2013.

\bibitem[Dominici et~al.(2006)Dominici, Peng, Bell, Pham, McDermott, Zeger, and Samet]{Dominici2006}
Francesca Dominici, Roger~D. Peng, Michelle~L. Bell, Luu Pham, Aidan McDermott, Scott~L. Zeger, and Jonathan~M. Samet.
\newblock {Fine Particulate Air Pollution and Hospital Admission for Cardiovascular and Respiratory Diseases}.
\newblock \emph{JAMA}, 295\penalty0 (10):\penalty0 1127, 2006.

\bibitem[Dorn and Guo(2023)]{Dorn2023}
Jacob Dorn and Kevin Guo.
\newblock {Sharp Sensitivity Analysis for Inverse Propensity Weighting via Quantile Balancing}.
\newblock \emph{Journal of the American Statistical Association}, 118\penalty0 (544):\penalty0 2645--2657, 2023.

\bibitem[Frauen et~al.(2023)Frauen, Melnychuk, and Feuerriegel]{Frauen2023}
Dennis Frauen, Valentyn Melnychuk, and Stefan Feuerriegel.
\newblock {Sharp Bounds for Generalized Causal Sensitivity Analysis}.
\newblock \emph{Advances in Neural Information Processing Systems}, 36, 2023.

\bibitem[Gibson et~al.(2022)Gibson, MacDonald, Fisher, Chen, Pawlick, and Cook]{Gibson2022}
Jacqueline~MacDonald Gibson, John~M. MacDonald, Michael Fisher, Xiwei Chen, Aralia Pawlick, and Philip~J. Cook.
\newblock {Early life lead exposure from private well water increases juvenile delinquency risk among US teens}.
\newblock \emph{Proceedings of the National Academy of Sciences}, 119\penalty0 (6):\penalty0 e2110694119, 2022.

\bibitem[Haneuse and Rotnitzky(2013)]{haneuse2013modified}
S.~Haneuse and A.~Rotnitzky.
\newblock Estimation of the effect of interventions that modify the received treatment.
\newblock \emph{Statistics in Medicine}, 32\penalty0 (30):\penalty0 5260--5277, 2013.

\bibitem[Hardle and Stoker(1989)]{Hardle1989}
Wolfgang Hardle and Thomas~M. Stoker.
\newblock {Investigating Smooth Multiple Regression by the Method of Average Derivatives}.
\newblock \emph{Journal of the American Statistical Association}, 84\penalty0 (408):\penalty0 986--995, 1989.

\bibitem[Hines et~al.(2023)Hines, Diaz-Ordaz, and Vansteelandt]{Hines2023OptimallyEffects}
Oliver Hines, Karla Diaz-Ordaz, and Stijn Vansteelandt.
\newblock {Optimally weighted average derivative effects}.
\newblock \emph{arXiv preprint arXiv:2308.05456}, 2023.

\bibitem[Hines and Miles(2025)]{hines2025learning}
Oliver~J Hines and Caleb~H Miles.
\newblock {Learning density ratios in causal inference using Bregman-Riesz regression}.
\newblock \emph{arXiv preprint arXiv:2510.16127}, 2025.

\bibitem[Hines et~al.(2026)Hines, Diaz-Ordaz, and Vansteelandt]{Hines2021ParameterisingEffects}
Oliver~J Hines, Karla Diaz-Ordaz, and Stijn Vansteelandt.
\newblock Parameterizing the effect of a continuous treatment using average derivative effects.
\newblock \emph{Biometrika}, 113\penalty0 (2):\penalty0 asag012, 2026.

\bibitem[Hsu and Small(2013)]{Hsu2013}
Jesse~Y. Hsu and Dylan~S. Small.
\newblock {Calibrating sensitivity analyses to observed covariates in observational studies}.
\newblock \emph{Biometrics}, 69\penalty0 (4):\penalty0 803--811, 2013.

\bibitem[Huang and Pimentel(2025)]{Huang2024}
Melody Huang and Samuel~D Pimentel.
\newblock {Variance-based sensitivity analysis for weighting estimators results in more informative bounds}.
\newblock \emph{Biometrika}, 112\penalty0 (1):\penalty0 asae040, 2025.

\bibitem[Imbens(2000)]{Imbens2000}
Guido Imbens.
\newblock {The Role of the Propensity Score in Estimating Dose-Response Functions}.
\newblock \emph{Biometrika}, 87\penalty0 (3):\penalty0 706--710, 2000.

\bibitem[Jesson et~al.(2022)Jesson, Douglas, Manshausen, Solal, Meinshausen, Stier, Gal, and Shalit]{Jesson2022}
Andrew Jesson, Alyson Douglas, Peter Manshausen, Maëlys Solal, Nicolai Meinshausen, Philip Stier, Yarin Gal, and Uri Shalit.
\newblock {Scalable Sensitivity and Uncertainty Analyses for Causal-Effect Estimates of Continuous-Valued Interventions}.
\newblock \emph{Advances in Neural Information Processing Systems}, 35, 2022.

\bibitem[Jin et~al.(2026)Jin, Ren, and Zhou]{Jin2022SensitivityPerspective}
Ying Jin, Zhimei Ren, and Zhengyuan Zhou.
\newblock Sensitivity analysis under the f-sensitivity model: A distributional robustness perspective.
\newblock \emph{Operations Research}, 74\penalty0 (2):\penalty0 860--878, 2026.

\bibitem[Kennedy(2019)]{Kennedy2019a}
Edward~H. Kennedy.
\newblock Nonparametric causal effects based on incremental propensity score interventions.
\newblock \emph{Journal of the American Statistical Association}, 114\penalty0 (526):\penalty0 645--656, 2019.

\bibitem[Kennedy(2024)]{Kennedy2022}
Edward~H Kennedy.
\newblock {Semiparametric doubly robust targeted double machine learning: a review}.
\newblock \emph{Handbook of Statistical Methods for Precision Medicine}, pages 207--236, 2024.

\bibitem[Kennedy et~al.(2017)Kennedy, Ma, Mchugh, and Small]{kennedy_dose_response}
Edward~H. Kennedy, Zongming Ma, Matthew~D. Mchugh, and Dylan~S. Small.
\newblock {Nonparametric methods for doubly robust estimation of continuous treatment effects}.
\newblock \emph{Journal of the Royal Statistical Society, Series B (Statistical Methodology)}, 79\penalty0 (4):\penalty0 1229--1245, 2017.

\bibitem[Klyne and Shah(2026)]{Klyne2023AverageLearning}
Harvey Klyne and Rajen~D Shah.
\newblock Average partial effect estimation using double machine learning.
\newblock \emph{The Annals of Statistics}, 54\penalty0 (1):\penalty0 176--200, 2026.

\bibitem[Lee et~al.(2026)Lee, Hubbard, and Schuler]{Lee2024}
Kaitlyn Lee, Alan Hubbard, and Alejandro Schuler.
\newblock Bridging binarization: causal inference with dichotomized continuous exposures.
\newblock \emph{Journal of Causal Inference}, 14\penalty0 (1):\penalty0 20240049, 2026.

\bibitem[Levis et~al.(2024)Levis, Kennedy, McClean, Balakrishnan, and Wasserman]{Levis2024}
Alexander~W Levis, Edward~H Kennedy, Alec McClean, Sivaraman Balakrishnan, and Larry Wasserman.
\newblock {Stochastic interventions, sensitivity analysis, and optimal transport}.
\newblock \emph{arXiv preprint arXiv:2411.14285}, 2024.

\bibitem[Levis et~al.(2025)Levis, Bonvini, Zeng, Keele, and Kennedy]{Levis2023}
Alexander~W Levis, Matteo Bonvini, Zhenghao Zeng, Luke Keele, and Edward~H Kennedy.
\newblock {Covariate-assisted bounds on causal effects with instrumental variables}.
\newblock \emph{Journal of the Royal Statistical Society Series B: Statistical Methodology}, page qkaf028, 2025.

\bibitem[Lundberg and Brand(2023)]{Lundberg2023}
Ian Lundberg and Jennie~E. Brand.
\newblock {The Nonlinear and Heterogeneous Effects of Parental Income on Children's Educational Attainment}.
\newblock \emph{OSF}, 2023.

\bibitem[Marmarelis et~al.(2023)Marmarelis, Jahanshad, Haddad, Galstyan, Jesson, and Steeg]{Marmarelis2023}
Myrl~G. Marmarelis, Neda Jahanshad, Elizabeth Haddad, Aram Galstyan, Andrew Jesson, and Greg~Ver Steeg.
\newblock {Partial Identification of Dose Responses with Hidden Confounders}.
\newblock \emph{Proceedings of Machine Learning Research}, 216\penalty0 (UAI):\penalty0 1368--1379, 2023.

\bibitem[McClean et~al.(2026)McClean, Branson, and Kennedy]{McClean2024a}
Alec McClean, Zach Branson, and Edward~H Kennedy.
\newblock Calibrated sensitivity models.
\newblock \emph{Biometrika}, page asag001, 2026.

\bibitem[Nabi et~al.(2024)Nabi, Bonvini, Kennedy, Huang, Smid, and Scharfstein]{Nabi2024}
Razieh Nabi, Matteo Bonvini, Edward~H Kennedy, Ming-Yueh Huang, Marcela Smid, and Daniel~O Scharfstein.
\newblock {Semiparametric sensitivity analysis: unmeasured confounding in observational studies}.
\newblock \emph{Biometrics}, 80\penalty0 (4):\penalty0 ujae106, 2024.

\bibitem[Newey and Stoker(1993)]{Newey1993EfficiencyModels}
Whitney~K. Newey and Thomas~M. Stoker.
\newblock {Efficiency of Weighted Average Derivative Estimators and Index Models}.
\newblock \emph{Econometrica}, 61\penalty0 (5):\penalty0 1199--1223, 1993.

\bibitem[Powell et~al.(1989)Powell, Stock, and Stoker]{Powell1989}
James~L. Powell, James~H. Stock, and Thomas~M. Stoker.
\newblock {Semiparametric Estimation of Index Coefficients}.
\newblock \emph{Econometrica}, 57\penalty0 (6):\penalty0 1403--1430, 1989.

\bibitem[Robins et~al.(2000)Robins, Rotnitzky, and Scharfstein]{Robins2000}
James~M. Robins, Andrea Rotnitzky, and Daniel~O. Scharfstein.
\newblock {Sensitivity Analysis for Selection Bias and Unmeasured Confounding in Missing Data and Causal inference Models}.
\newblock In \emph{Statistical Models in Epidemiology, the Environment, and Clinical Trials}, pages 1--94, New York, NY, 2000. Springer New York.

\bibitem[Rosenbaum(1989)]{rosenbaum1989sensitivity}
P.~R. Rosenbaum.
\newblock {Sensitivity Analysis for Matched Observational Studies with Many Ordered Treatments}.
\newblock \emph{Scandinavian Journal of Statistics}, 16\penalty0 (3):\penalty0 227--236, 1989.

\bibitem[Rosenbaum(2002)]{rosenbaum_obs}
Paul~R. Rosenbaum.
\newblock \emph{{Observational Studies}}.
\newblock Springer New York, 2002.

\bibitem[Rothenh{\"{a}}usler and Yu(2019)]{Rothenhausler2019}
Dominik Rothenh{\"{a}}usler and Bin Yu.
\newblock {Incremental causal effects}.
\newblock \emph{arXiv preprint arXiv:1907.13258}, 2019.

\bibitem[Schindl et~al.(2024)Schindl, Shen, and Kennedy]{Schindl2024}
Kyle Schindl, Shuying Shen, and Edward~H Kennedy.
\newblock {Incremental effects for continuous exposures}.
\newblock \emph{arXiv preprint arXiv:2409.11967}, 2024.

\bibitem[Stoker(1986)]{Stoker1986}
Thomas~M Stoker.
\newblock {Consistent Estimation of Scaled Coefficients}.
\newblock \emph{Econometrica}, 54\penalty0 (6):\penalty0 1461--1481, 1986.

\bibitem[Tan(2006)]{Tan2006AScores}
Zhiqiang Tan.
\newblock A distributional approach for causal inference using propensity scores.
\newblock \emph{Journal of the American Statistical Association}, 101\penalty0 (476):\penalty0 1619--1637, 2006.

\bibitem[Tan(2025)]{Tan2025}
Zhiqiang Tan.
\newblock {Sensitivity models and bounds under sequential unmeasured confounding in longitudinal studies}.
\newblock \emph{Biometrika}, 112\penalty0 (1), 2025.

\bibitem[Tsiatis(2006)]{Tsiatis2006}
Anastasios~A. Tsiatis.
\newblock \emph{{Semiparametric Theory and Missing Data}}.
\newblock Springer Series in Statistics. Springer New York, New York, NY, 2006.

\bibitem[VanderWeele and Hernan(2013)]{VanderWeele2013}
Tyler~J. VanderWeele and Miguel~A. Hernan.
\newblock {Causal inference under multiple versions of treatment}.
\newblock \emph{Journal of Causal Inference}, 1\penalty0 (1):\penalty0 1--20, 2013.

\bibitem[Yadlowsky et~al.(2022)Yadlowsky, Namkoong, Basu, Duchi, and Tian]{Yadlowsky2022}
Steve Yadlowsky, Hongseok Namkoong, Sanjay Basu, John Duchi, and Lu~Tian.
\newblock {Bounds on the conditional and average treatment effect with unobserved confounding factors}.
\newblock \emph{The Annals of Statistics}, 50\penalty0 (5):\penalty0 2587--2615, 2022.

\bibitem[Zhang and Heng(2025)]{Zhang2024b}
Jeffrey Zhang and Siyu Heng.
\newblock Bridging the gap between design and analysis: randomization inference and sensitivity analysis for matched observational studies with treatment doses.
\newblock \emph{Biometrics}, 81\penalty0 (4):\penalty0 ujaf156, 2025.

\bibitem[Zhang et~al.(2024)Zhang, Small, and Heng]{Zhang2024a}
Jeffrey Zhang, Dylan~S Small, and Siyu Heng.
\newblock {Sensitivity analysis for matched observational studies with continuous exposures and binary outcomes}.
\newblock \emph{Biometrika}, 111\penalty0 (4):\penalty0 1349--1368, 2024.

\bibitem[Zhang and Zhao(2026)]{Zhang2022a}
Yao Zhang and Qingyuan Zhao.
\newblock An average-case sensitivity analysis for unmeasured confounding.
\newblock \emph{Biometrika}, 113\penalty0 (2):\penalty0 asag030, 2026.

\bibitem[Zhao et~al.(2019)Zhao, Small, and Bhattacharya]{Zhao2019SensitivityBootstrap}
Qingyuan Zhao, Dylan~S Small, and Bhaswar~B Bhattacharya.
\newblock {Sensitivity analysis for inverse probability weighting estimators via the percentile bootstrap}.
\newblock \emph{Journal of the Royal Statistical Society Series B: Statistical Methodology}, 81\penalty0 (4):\penalty0 735--761, 2019.

\bibitem[Zheng et~al.(2025)Zheng, D'Amour, and Franks]{zheng2025copula}
Jiajing Zheng, Alexander D'Amour, and Alexander Franks.
\newblock Copula-based sensitivity analysis for multi-treatment causal inference with unobserved confounding.
\newblock \emph{Journal of Machine Learning Research}, 26\penalty0 (36):\penalty0 1--60, 2025.

\end{thebibliography}

\newpage
\appendix

\section{Proof of results in Section \ref{section: prelim}}
\label{sec: estimand appendix}
\begin{proof}[Proof of Lemma \ref{prop: binary estimand}]
We first show the second and third equalities. Consider $(a,x,u)$ such that $f(a,x,u) > 0$. By Assumptions \ref{assumption: latent ignorability} and \ref{assumption: SUTVA},
\begin{align*}
E[Y(a) \mid x, u] &= E[Y(a) \mid a, x, u] \\ &= E[Y \mid a, x, u].
\end{align*}
The second equality then follows by taking derivatives on both sides and then integrating over the support of $(A, X, U)$. The third equality follows by an integration by parts argument that is permitted by Assumption \ref{assumption: regularity} \citep{Powell1989}. Now we turn to the first equality. Observe that
\begin{equation*}
\begin{aligned}
\lim_{\delta \to 0 }\delta^{-1} E[Y(A+\delta) - Y(A)] &= \lim_{\delta \to 0 }\delta^{-1} \int E[Y(a+\delta) - Y(a) \mid a] f(a) da \\ &= \lim_{\delta \to 0 }\delta^{-1} \int E[Y(a+\delta) - Y(a) \mid x, u, a] f(x, u \mid a) dudx f(a) da \\ &= \lim_{\delta \to 0 }\delta^{-1} \int E[Y(a+\delta) - Y(a) \mid x, u] f(a, x, u) da dx du \\ &=  \int \lim_{\delta \to 0 }\delta^{-1} E[Y(a+\delta) - Y(a) \mid x, u] f(a, x, u) da dx du \\ &=  E[\partial_a E[Y(A)\mid X, U]].
\end{aligned}
\end{equation*}
The first equality is by definition, the second is by iterated expectation, and the third by Assumption \ref{assumption: latent ignorability}. The fourth equality follows from Assumption \ref{assumption: regularity} (continuity and boundedness of $\partial_a E[Y(a)\mid x, u]$) and the dominated convergence theorem.
\end{proof}
\cite{Rothenhausler2019} introduce a causal estimand for continuous outcomes called the incremental effect. They introduce the following assumption on the potential outcomes.
\begin{assumption}[Regularity - potential outcomes]
\label{assumption: regularity - potential outcomes}
The potential outcomes $Y(a)$ are bounded and the derivative $Y'(a) := \partial_a Y(a) = \lim_{\delta \to 0}\delta^{-1} [Y(a+\delta) - Y(a)]$ is continuous and bounded.
\end{assumption}
The estimand of interest is the incremental effect, $E[Y'(A)]$, where the average is taken over potential outcomes, treatments, and confounders. Combining the arguments in Proposition 1 of \cite{Rothenhausler2019} and the derivation of Lemma \ref{prop: binary estimand}, it is straightforward to check that $\theta$ as defined in \eqref{eq: ADE definition} is equal to $E[Y'(A)]$. Thus, when the outcome is continuous and one is willing to invoke Assumption \ref{assumption: regularity - potential outcomes}, our methods are directly applicable to $E[Y'(A)]$.

\section{Proof of results in Section \ref{section: sensitivity model}}

\begin{proof}[Proof of Lemma \ref{lemma: model implication on latent score}]
Consider model \eqref{eqn:sens_model} first.
First, taking logs in \eqref{eqn:sens_model}, we get 
\begin{equation*}
-\gamma |a - a'| \leq \log(f(a \mid x)) - 
 \log(f(a' \mid x)) - 
 \{ \log(f(a \mid x,u )) - \log(f(a' \mid x, u))\} \leq \gamma |a - a'|.
\end{equation*}
Next, recall that $s(a \mid x,u) = {\partial_a} \log(f(a\mid x,u)) = \lim_{h \to 0} \frac{\log(f(a + h\mid x,u)) - \log(f(a\mid x,u))}{h}$. By the above equation, plugging in $a$ for $a$ and $a+h$ for $a'$, adding the $\log(f(a + h\mid x)) - \log(f(a\mid x))$ term everywhere in the inequality and dividing by $h$, we get 
\begin{equation*}
\begin{aligned}
    \frac{\log(f(a + h\mid x)) - \log(f(a\mid x))-\gamma h}{h} &\leq \frac{\log(f(a + h\mid x,u)) - \log(f(a\mid x,u))}{h} \\ &\leq \frac{\log(f(a + h\mid x)) - \log(f(a\mid x))+\gamma h}{h}, 
\end{aligned}
\end{equation*}
which immediately implies the result after taking the limit $h \to 0$ everywhere. We can also show the implication under the Rosenbaum-style sensitivity model introduced later in \eqref{eqn:sens_model 1}. Taking logs, we get 
\begin{equation*}
-\gamma |a - a'| \leq \log(f(a \mid x, u')) - 
 \log(f(a' \mid x, u')) - 
 \{ \log(f(a \mid x,u )) - \log(f(a' \mid x, u))\} \leq \gamma |a - a'|.
\end{equation*}
Using an analogous argument, we can deduce that
\begin{equation*}
s(a \mid x, u') -\gamma \leq s(a \mid x, u) \leq s(a \mid x, u') + \gamma \ \forall a, x, u, u'.
\end{equation*}
The result follows from an application of the upcoming Lemma \ref{lemma: latent score integrates to marginal score} combined with the following fact: For a random variable with the property that any two points in its support lie within $\gamma$ of each other, it must be the case that any point of its support must lie within $\gamma$ of the mean of the random variable. 
\end{proof}

\begin{proof}[Proof of Lemma \ref{lemma: latent score integrates to marginal score}]
A similar argument appears in Proposition 2 of \cite{Rothenhausler2019}, but we provide one for completeness. By definition, we have
\begin{equation*}
\begin{aligned}
\int s(a \mid x, u) f(u \mid a, x) du &= \int f'(a \mid x, u)/f(a \mid x,u) f(u \mid a, x) du \\ &= \int f'(a \mid x, u)/ f(a \mid  x) f(u \mid x) du \\ &= 1/f(a \mid  x)\int \lim_{h \to 0}\frac{f(a+h \mid x, u)-f(a \mid x, u)}{h}f(u \mid x) du \\ &= 1/f(a \mid  x)  \lim_{h \to 0}1/h\int \{f(a+h \mid x, u)-f(a \mid x, u)\}f(u \mid x) du \\ &= 1/f(a \mid  x)  \lim_{h \to 0}1/h\int f(a+h, u \mid x)-f(a, u \mid x)du \\ &= 1/f(a \mid  x)  \lim_{h \to 0}1/h \{f(a+h \mid x)-f(a\mid x)\} \\ &= f'(a\mid x)/f(a \mid x) = s(a \mid x).
\end{aligned}
\end{equation*}
The second equality is by Bayes. The fourth equality holds by invoking the dominated convergence theorem (which is possible by the boundedness and continuity assumption), which allows for interchange of limit and integration. The remaining equalities are algebraic or by definition. 
\end{proof}

\section{Proof of results in Section \ref{section: optimization}}
\begin{proof}[Proof of Proposition \ref{prop: solution continuous}]
We first characterize the Lagrangian of the minimization program. It is as follows:
\begin{equation*}
\begin{aligned}
    \mathcal{L} &= E_{Y,U \mid A, X}\{-s(A \mid X, U)Y + \lambda_1 (s(A \mid X)- s(A \mid X, U) - \gamma) \\ &+ \lambda_2 (s(A \mid X, U)- s(A \mid X) - \gamma)+ \lambda_3(s(A \mid X)-E_{Y,U \mid A, X}[s(A \mid X,U)])\}. 
\end{aligned}
\end{equation*}
Taking the derivative of the Lagrangian with respect to $s(A\mid X, U)$, we get that
\begin{equation*}
    -Y - \lambda_1 +\lambda_2 - \lambda_3 = 0.
\end{equation*}
By complementary slackness, we know
\begin{equation*}
    \lambda_1 (s(A \mid X)- s(A \mid X, U) - \gamma) = 0 \text{ and } \lambda_2 (s(A \mid X, U)- s(A \mid X) - \gamma) = 0.
\end{equation*}
We also know that at the optimum, $\lambda_1, \lambda_2 \geq 0$. If $-\lambda_3 - Y < 0$, then it must be the case that $\lambda_1 < \lambda_2$, in which case $\lambda_1 = 0$ and $s(A \mid X, U)= s(A \mid X) + \gamma$. Similarly, if $-\lambda_3 - Y >0$, then it must be the case that $\lambda_1 > \lambda_2$, in which case $\lambda_2 = 0$ and $s(A \mid X, U)= s(A \mid X) - \gamma$. In the primal problem, we also have the constraint
\begin{equation*}
    E_{Y,U \mid A,X}[s(A \mid X,U)] = s(A \mid X).
\end{equation*}
Let $\alpha^* = P(Y < -\lambda_3 \mid A, X)$. It then follows that \begin{equation*}
    E_{Y,U \mid A,X}[s(A \mid X,U)] = \alpha^*(s(A \mid X) - \gamma) + (1-\alpha^*)(s(A \mid X) + \gamma).
\end{equation*}
Solving for $\alpha^*$, we get that $\alpha^* = 1/2$. Thus, $-\lambda_3$ is simply the median of $Y \mid A, X$. This means that the (minimization) optimization in \eqref{eq: optimization problem} is solved by 
\begin{equation}
    s^{*}_{\text{min}}(A \mid X,U) = \begin{cases}
        s(A \mid X) - \gamma & \text{if } Y < \text{median}(Y \mid A,X)\\
        s(A \mid X) + \gamma & \text{if } Y > \text{median}(Y \mid A,X)
    \end{cases}.
\end{equation}
It is the opposite for the maximization, i.e.
\begin{equation}
    s^{*}_{\text{max}}(A \mid X,U) = \begin{cases}
        s(A \mid X) + \gamma & \text{if } Y < \text{median}(Y \mid A,X)\\
        s(A \mid X) - \gamma & \text{if } Y > \text{median}(Y \mid A,X)
    \end{cases}.
\end{equation}
This result can be derived by replacing $-Y$ in the Lagrangian with $Y$. Next, the optimal (maximum) value achieved by the solution is the following:
\begin{equation}
    \psi_{\max} = \psi^+ +\psi^-,
\end{equation}
where $\psi^+ = E[-(s(A \mid X) - \gamma) \mu^+(A,X)]$ and $\psi^- = E[-(s(A \mid X) + \gamma) \mu^-(A, X)]$, where $\mu^+(A,X) = E[Y \mathbbm{1}_{\{Y > M( A, X)\}}\mid A, X]$ and $\mu^-(A,X) = E[Y \mathbbm{1}_{\{Y < M( A, X)\}} \mid A, X]$. A simple application of iterated expectation yields $\psi_{\max} = E[-s(A \mid X)Y] + \gamma E[Y (\mathbbm{1}_{\{Y > M( A, X)\}} - \mathbbm{1}_{\{Y < M( A, X)\}})]$. For the minimum, it is easy to see that $\psi_{\min} = E[-s(A \mid X)Y] - \gamma E[Y (\mathbbm{1}_{\{Y > M( A, X)\}} - \mathbbm{1}_{\{Y < M( A, X)\}})]$, as desired.
\end{proof}

\begin{proof}[Proof of Proposition \ref{prop: solution binary}]
Since $Y$ is binary, there must exist an optimal solution that sets $ s^*(A \mid X, U)$ to distinct values according to whether $Y = 1$ or $Y = 0$. Since the optimization is conditional on $A, X$, we may rewrite the optimization as 
\begin{equation}
\label{eq: binary optimization rewritten}
\begin{aligned}
& \underset{s^1(a\mid x, u), s^0(a \mid x,u)}{\text{maximize/minimize}}
& & E[-s^1(A \mid X, U) Y \mid A = a, X = x] \\
& \text{subject to} & & s^1(a \mid x, u) \in [s(a\mid x) - \gamma, s(a \mid x) + \gamma] \ \forall a, x.\\ & \text{and} & & s^0(a \mid x, u) \in [s(a\mid x) - \gamma, s(a \mid x) + \gamma] \ \forall a, x.\\
& \text{and} & & s^1(A \mid X, U) P(Y = 1 \mid A = a, X = x) + \\ & & & s^0(A \mid X, U) (1 - P(Y = 1 \mid A = a, X = x)) = s(a \mid x).
\end{aligned}
\end{equation}

The objective function can be written as $-s^1(A \mid X, U) P(Y = 1 \mid A = a, X = x)$. To maximize, we simply take $s^1(A \mid X, U)$ to be the smallest possible it can be while still satisfying all three constraints in \eqref{eq: binary optimization rewritten}. It is clear that if $P(Y = 1 \mid A, X) \leq 1/2$, we can take $s^1(A \mid X, U) = s(a \mid x) - \gamma$, and  $s^0(A \mid X, U) =  s(A \mid X) + \gamma \times P(Y = 1 \mid A, X)/(1-P(Y = 1 \mid A, X)) \leq s(a\mid x) + \gamma$ since $P(Y = 1 \mid A, X)/(1-P(Y = 1 \mid A, X)) \leq 1$. It is easy to check that the third constraint holds. Conversely, if $P(Y = 1 \mid A, X) \geq 1/2$, we can only take $s^1(A \mid X, U) = s(a \mid x) - \gamma \times (1-P(Y = 1 \mid A, X))/ P(Y = 1 \mid A, X) \geq s(a \mid x) - \gamma$ as making it any smaller will cause violation of one of the second or third constraints in \eqref{eq: binary optimization rewritten}. To allow $s^1(A, \mid X, U)$ to be that small, we must take $s^0(A \mid X, U) =  s(A \mid X) + \gamma $ to be the maximum. It is easy to check that the third constraint holds under this choice. It follows that the maximum value takes the form 
\begin{align*}
\psi_{\max}^B &=  E[-(s(A \mid X) - \gamma) \mathbbm{1}(P(Y = 1\mid A,X) \leq 1/2)P(Y = 1 \mid A, X)] \\ &+ E[-(s(A \mid X) - \gamma \times \frac{P(Y = 0 \mid A, X)}{P(Y = 1 \mid A, X)}) \mathbbm{1}(P(Y = 1\mid A,X) > 1/2)P(Y = 1 \mid A, X)].
\end{align*} 
We can simplify this to 
\begin{align*}
\psi_{\max}^B &= E[-s(A\mid X)P(Y = 1 \mid A,X)]+ \gamma E[\mathbbm{1}(P(Y = 1\mid A,X) \leq 1/2)P(Y = 1 \mid A, X)\\ &+\mathbbm{1}(P(Y = 1\mid A,X) > 1/2)P(Y = 0 \mid A, X)] \\&= E[-s(A\mid X)Y] + \gamma E[1/2 - |P(Y = 1 \mid A,X) - 1/2|].
\end{align*}
 The derivation for the minimum follows analogously.
\end{proof}
\section{Deriving the efficient influence functions}
\subsection{Continuous outcome}
We derive the efficient influence function of the correction term in the minimal value of the ADE under the sensitivity model at a fixed $\gamma > 0$. In some parts, the steps resemble an argument from \cite{Zhang2022a}. We consider a parametric submodel indexed by $\epsilon$ that passes through the truth at $\epsilon = 0$. We define score functions $S$ with respect to the parametric submodel so that $S(y,a,x) = S(y \mid a,x) + S(a,x)$, where $S(y \mid a,x) = \frac{\partial }{\partial \epsilon}\log(f_{\epsilon}(y \mid a,x))\rvert_{\epsilon = 0}$ and $S(a,x) = \frac{\partial }{\partial \epsilon}\log(f_{\epsilon}(a,x))\rvert_{\epsilon = 0}$. First, we prove a useful lemma.
\begin{lemma}
\begin{align*}
&\frac{d}{d\epsilon} E_\epsilon[Y \mathbbm{1}_{\{Y 
\leq M_\epsilon( A, X)\}}\mid A, X] \rvert_{\epsilon = 0} =  E[(Y - M(A,X))\mathbbm{1}_{\{Y \leq M(A,X)\}}S(Y \mid A, X) \mid A, X], \\ & \frac{d}{d\epsilon} E_\epsilon[Y \mathbbm{1}_{\{Y 
>M_\epsilon( A, X)\}}\mid A, X] \rvert_{\epsilon = 0} =  E[(Y - M(A,X))\mathbbm{1}_{\{Y > M(A,X)\}}S(Y \mid A, X) \mid A, X].
\end{align*}
\end{lemma}
\begin{proof}
We first show the first equality. By definition of $M(A,X)$ as the conditional median of $Y \mid A,X$, and then taking derivatives using Leibniz rule,
\begin{equation*}
\begin{aligned}
    1/2 &= \int^{M_\epsilon(A,X)}p_\epsilon(y \mid A,X) dy \implies \\ 0 &= p(y = M(A,X) \mid A, X)\frac{d}{d\epsilon}M_\epsilon(A,X)\rvert_{\epsilon = 0} + E[\mathbbm{1}_{\{Y \leq M(A,X)\}}S(Y \mid A, X) \mid A, X] \implies \\  &p(y = M(A,X) \mid A, X)\frac{d}{d\epsilon}M_\epsilon(A,X)\rvert_{\epsilon = 0} = -E[\mathbbm{1}_{\{Y \leq M(A,X)\}}S(Y \mid A, X) \mid A, X].
\end{aligned}
\end{equation*}
Again using Leibniz rule and the above equation, 
\begin{equation*}
\begin{aligned}
\frac{d}{d\epsilon} &E_\epsilon[Y \mathbbm{1}_{\{Y \leq M_\epsilon( A, X)\}}\mid A, X] \rvert_{\epsilon = 0} = \frac{d}{d\epsilon}\int^{M_\epsilon(A,X)}y p_\epsilon(y \mid A,X) dy \rvert_{\epsilon = 0}  \\  &= M(A,X)p(y = M(A,X) \mid A,X)\frac{d}{d\epsilon}M_\epsilon(A,X)\rvert_{\epsilon = 0} + E[Y\mathbbm{1}_{\{Y \leq M(A,X)\}}S(Y \mid A, X) \mid A, X] \\ &= -M(A,X)E[\mathbbm{1}_{\{Y \leq M(A,X)\}}S(Y \mid A, X) \mid A, X] + E[Y\mathbbm{1}_{\{Y \leq M(A,X)\}}S(Y \mid A, X) \mid A, X] \\ &= E[(Y - M(A,X))\mathbbm{1}_{\{Y \leq M(A,X)\}}S(Y \mid A, X) \mid A, X].
\end{aligned}
\end{equation*}
The second equality follows from a similar argument. By definition of $M(A,X)$ as the conditional median of $Y \mid A,X$, and then taking derivatives using Leibniz rule,
\begin{equation*}
\begin{aligned}
    1/2 &= \int_{M_\epsilon(A,X)}p_\epsilon(y \mid A,X) dy \implies \\ 0 &= -p(y = M(A,X) \mid A, X)\frac{d}{d\epsilon}M_\epsilon(A,X)\rvert_{\epsilon = 0} + E[\mathbbm{1}_{\{Y > M(A,X)\}}S(Y \mid A, X) \mid A, X] \implies \\  & -p(y = M(A,X) \mid A, X)\frac{d}{d\epsilon}M_\epsilon(A,X)\rvert_{\epsilon = 0} = -E[\mathbbm{1}_{\{Y > M(A,X)\}}S(Y \mid A, X) \mid A, X].
\end{aligned}
\end{equation*}
Again using Leibniz rule and the above equation, 
\begin{equation*}
\begin{aligned}
\frac{d}{d\epsilon} &E_\epsilon[Y \mathbbm{1}_{\{Y > M_\epsilon(A, X)\}}\mid A, X] \rvert_{\epsilon = 0} = \frac{d}{d\epsilon}\int_{M_\epsilon(A,X)}y p_\epsilon(y \mid A,X) dy \rvert_{\epsilon = 0}  \\  &= -M(A,X)p(y = M(A,X) \mid A,X)\frac{d}{d\epsilon}M_\epsilon(A,X)\rvert_{\epsilon = 0} + E[Y\mathbbm{1}_{\{Y > M(A,X)\}}S(Y \mid A, X) \mid A, X] \\ &= -M(A,X)E[\mathbbm{1}_{\{Y > M(A,X)\}}S(Y \mid A, X) \mid A, X] + E[Y\mathbbm{1}_{\{Y > M(A,X)\}}S(Y \mid A, X) \mid A, X] \\ &= E[(Y - M(A,X))\mathbbm{1}_{\{Y > M(A,X)\}}S(Y \mid A, X) \mid A, X].
\end{aligned}
\end{equation*}
\end{proof}
Building off of the previous lemma, the following result derives the efficient influence functions of two quantities whose difference (times $\gamma$) is exactly the functional introduced in Proposition \ref{prop: eif correction continuous}. Proposition \ref{prop: eif correction continuous} is then an immediate corollary after taking the difference of the two efficient influence functions. 
\begin{lemma}
The efficient influence functions of $\theta^+ \equiv  E[Y \mathbbm{1}_{\{Y > M( A, X)\}}]$ and \\ $\theta^- \equiv E[Y \mathbbm{1}_{\{Y \leq M( A, X)\}}]$ in the nonparametric model are
\begin{equation*}
\begin{aligned}
(Y - M(A,X))\mathbbm{1}_{\{Y 
> M(A,X)\}} +1/2M(A,X) \text{and } \\ \ (Y - M(A,X))\mathbbm{1}_{\{Y \leq M(A,X)\}}+1/2M(A,X),
\end{aligned}
\end{equation*}
respectively.
\end{lemma}
\begin{proof}
We start with $\theta^-$. The goal is to find a random variable $\phi^-$ that is a function of the data such that $\frac{d}{d\epsilon} \theta_{\epsilon}^-\rvert_{\epsilon = 0} = E[\phi^-(Y,A,X) S(Y,A,X)]$, where $S(y,a,x) = \frac{d}{d\epsilon} \log(f_{\epsilon}(y,a,x)) \rvert_{\epsilon = 0}$. We start by taking derivatives: 
\begin{equation*}
\begin{aligned}
\frac{d}{d\epsilon} &\theta_{\epsilon}^-\rvert_{\epsilon = 0}= \frac{d}{d\epsilon} E_{\epsilon}[\mu_{\epsilon}^-(A,X) ] \rvert_{\epsilon = 0}\\ &= \frac{d}{d\epsilon} \int  \mu_{\epsilon}^-(a,x)f_{\epsilon}(a,x) da dx \rvert_{\epsilon = 0}\\ &=  \underbrace{\frac{d}{d\epsilon} \int\mu^-(a,x)f_{\epsilon}(a,x) da dx \rvert_{\epsilon = 0}}_\text{I} + \underbrace{\frac{d}{d\epsilon} \int  \mu_{\epsilon}^-(a,x)f(a,x) da dx\rvert_{\epsilon = 0}}_\text{II} 
\end{aligned}
\end{equation*}
We deal with each of the terms separately.

Term I: 
\begin{equation*}
\begin{aligned}
\frac{d}{d\epsilon} &\int  \mu^-(a,x)f_{\epsilon}(a,x) dadx \rvert_{\epsilon = 0} = E[\mu^-(A,X) S(A,X)] \\ &= E[\mu^-(A,X) S(A,X)] + E[\mu^-(A,X) S(Y \mid A,X)] \\ &= E[\mu^-(A,X) S(Y, A,X)],
\end{aligned}
\end{equation*}
The second equality holds because by a property of conditional scores, $S(Y \mid A,X)$ is mean zero conditional on $(A,X)$.

Term II: 
\begin{equation*}
\begin{aligned}
\frac{d}{d\epsilon} &\int  \mu_{\epsilon}^-(a,x)f(a,x) dadx \rvert_{\epsilon = 0} = \int  \frac{d}{d\epsilon} \mu_{\epsilon}^-(a,x)\rvert_{\epsilon = 0}f(a,x) dadx \\ &= \int E[(Y - M(a,x))\mathbbm{1}_{\{Y \leq M(a,x)\}}S(Y \mid a,x) \mid a,x]f(a,x) dadx \\ &= E[E[(Y - M(A,X))\mathbbm{1}_{\{Y \leq M(A,X)\}}S(Y \mid A,X) \mid A,X]] \\ &= E[(Y - M(A,X))\mathbbm{1}_{\{Y \leq M(A,X)\}}S(Y \mid A,X)],
\end{aligned}
\end{equation*}
where the second equality follows from the above lemma. Next, observe that $E[1/2 M(A,X) S(Y \mid A,X)] = 0$ and $E[\mu^-(A,X) S(Y \mid A,X)] = 0$ again because of the property that $S(Y \mid A,X)$ is mean zero given $(A,X)$. Thus,
\begin{equation*}
\begin{aligned}
&E[(Y - M(A,X))\mathbbm{1}_{\{Y \leq M(A,X)\}}S(Y \mid A,X)] \\ &= E[[(Y - M(A,X))\mathbbm{1}_{\{Y \leq M(A,X)\}} +1/2M(A,X) - \mu^-(A,X)]S(Y \mid A,X)] \\ &= E[[(Y - M(A,X))\mathbbm{1}_{\{Y \leq M(A,X)\}} +1/2M(A,X) - \mu^-(A,X)]S(Y \mid A,X)] \\ &+ E[[(Y - M(A,X))\mathbbm{1}_{\{Y \leq M(A,X)\}} +1/2M(A,X) - \mu^-(A,X)]S(A,X)] \\ &= E[[(Y - M(A,X))\mathbbm{1}_{\{Y \leq M(A,X)\}} +1/2M(A,X) - \mu^-(A,X)]S(Y, A,X)],
\end{aligned}
\end{equation*}
where the second to last equality holds because $(Y - M(A,X))\mathbbm{1}_{\{Y \leq M(A,X)\}} +1/2M(A,X) - \mu^-(A,X)$ is mean zero given $A,X$. To see why, observe that $E[Y \mathbbm{1}_{\{Y \leq M(A,X)\}} \mid A,X] = \mu^-(A,X)$ by definition, and $E[M(A,X)\mathbbm{1}_{\{Y \leq M(A,X)\}} \mid A,X] = 1/2 M(A,X)$ by definition of the conditional median $M(A,X)$. 

Putting Term I and Term II together, we get the uncentered EIF of $\theta^-$ as 
\begin{align*}
    (Y - M(A,X))\mathbbm{1}_{\{Y \leq M(A,X)\}} &+1/2M(A,X) - \mu^-(A,X) + \mu^-(A,X) = \\ &(Y - M(A,X))\mathbbm{1}_{\{Y \leq M(A,X)\}} +1/2M(A,X).
\end{align*}
The steps for $\theta^+$ are virtually identical to those for $\theta^-$, and thus the argument is omitted.
\end{proof}

\subsection{Binary outcome - smooth functional}
For this proof, we introduce some additional notation. Let $\mu(a,x) = P(Y = 1 \mid a, x)$. In addition, compared to the previous section, the score function $S(Y \mid A,X)$ changes slightly. Namely, since $Y$ is binary, we have that 
\begin{equation}
\label{eq: binary score}
S(y \mid a, x) = \frac{\dot{\mu}(a,x)(y - \mu(a,x))}{\mu(a,x)(1-\mu(a,x))},
\end{equation}
where $\dot{\mu}(a,x) \equiv \frac{\partial}{\partial \epsilon} P_\epsilon(Y = 1 \mid a, x)|_{\epsilon = 0}$. We also introduce the following lemma:
\begin{lemma} Let $h$ be a differentiable, scalar-valued function. Then
\begin{equation*}
    \frac{d}{d\epsilon}h\{P_{\epsilon}(Y = 1 \mid a,x)\}\rvert_{\epsilon = 0} = h'(\mu(a,x)) S(y \mid a, x) \mu(a,x)(1-\mu(a,x)) / (y - \mu(a,x)).
\end{equation*}
\end{lemma}
\begin{proof}
By the chain rule and evaluating at $\epsilon = 0$,
\begin{equation*}
\frac{d}{d\epsilon}h\{P_{\epsilon}(Y = 1 \mid a,x)\}\rvert_{\epsilon = 0} = h'(\mu(a,x)) \dot{\mu}(a,x).
\end{equation*}
Rearranging Equation \eqref{eq: binary score} and substituting into the previous equation yields the result.
\end{proof}

We can now prove Proposition \ref{prop: eif correction binary}.
\begin{proof}
 Again, the goal is to find a function $\phi_B$ such that $\frac{d \theta^B_{\epsilon}}{d \epsilon} |_{\epsilon = 0} = E[\phi^B(Y,A,X) S(Y,A,X)]$. We start by taking derivatives:
\begin{align*}
\frac{d \theta^B_{\epsilon}}{d \epsilon} \rvert_{\epsilon = 0} &= \frac{d }{d \epsilon}  E_{\epsilon}[h\{P_{\epsilon}(Y = 1 \mid A,X)\}] \rvert_{\epsilon = 0}\\ &= \frac{d }{d \epsilon} \int h\{P_{\epsilon}(Y = 1 \mid a,x)\} f_{\epsilon}(a,x) da dx \rvert_{\epsilon = 0} \\ &= \underbrace{ \int \frac{d }{d \epsilon} \{h\{P_{\epsilon}(Y = 1 \mid a,x)\} \rvert_{\epsilon = 0}f(a,x) da dx}_\text{I} \\&+ \underbrace{ \int  h\{P(Y = 1 \mid a,x)\}\frac{d}{d\epsilon}f_\epsilon(a,x)\rvert_{\epsilon = 0} da dx}_\text{II}.
\end{align*}
We will deal with the terms separately. 

For Term I:
\begin{align*}
& \int \frac{d }{d \epsilon} \{h\{P_{\epsilon}(Y = 1 \mid a,x)\} \rvert_{\epsilon = 0}f(y \mid a,x) f(a,x)  dy da dx \\ &=  \int h'(\mu(a,x))S(y \mid a, x) \mu(a,x)(1-\mu(a,x)) / (y - \mu(a,x)) f(y \mid a,x) f(a,x) dy da dx.
\end{align*}
Now, we compute the conditional mean on $(a,x)$ of $$ h'(\mu(a,x)) \mu(a,x)(1-\mu(a,x)) / (y - \mu(a,x)).$$  Since $y$ is binary with success probability $\mu(a,x)$, this quantity is 
\begin{equation*}
h'(\mu(a,x)) \mu(a,x)^2 - h'(\mu(a,x)) (1-\mu(a,x))^2.
\end{equation*}
We continue where we left off,
\begin{align*}
 & \int h'\{\mu(a,x)\} S(y \mid a, x) \mu(a,x)(1-\mu(a,x)) / (y - \mu(a,x)) f(y \mid a,x) f(a,x) dy da dx \\ &= E[h'\{\mu(A,X)\} S(Y \mid A, X) \mu(A,X)(1-\mu(A,X)) / (Y - \mu(A,X))] \\ &= E[h'\{\mu(A,X)\} S(Y \mid A, X) \mu(A,X)(1-\mu(A,X)) / (Y - \mu(A,X))]\\ &-(  E[\{h'(\mu(A,X)) \mu(A,X)^2 - h'(\mu(A,X)) (1-\mu(A,X))^2\} \{S(Y \mid A, X)\} ] ) \\ &=  E[h'(\mu(A,X))  \{Y \mu(A,X) - (1-Y)(1-\mu(A,X) \}S(Y \mid A, X)] \\ &-(  E[\{h'(\mu(A,X)) \mu(A,X)^2 - h'(\mu(A,X)) (1-\mu(A,X))^2\} \{S(Y \mid A, X)\} ] ) \\ &=  E[h'(\mu(A,X))  \{Y \mu(A,X) - (1-Y)(1-\mu(A,X) \}S(Y \mid A, X)] \\ &-(  E[\{h'(\mu(A,X)) (2\mu(A,X) - 1)\} \{S(Y \mid A, X)\} ] ) \\ &= E[h'\{\mu(A,X)\} \{Y  - \mu(A,X) \}S(Y \mid A, X)] \\ &= E[h'\{\mu(A,X)\} \{Y  - \mu(A,X) \}S(Y \mid A, X)] \\&-E[h'\{\mu(A,X)\} \{Y  - \mu(A,X) \}S(A, X)] \\ &= E[h'\{\mu(A,X)\} \{Y  - \mu(A,X) \}S(Y, A, X)].
\end{align*}
The second equality is by the fact that $S(Y \mid A,X)$ is mean zero given $(A,X)$. The third equality is due to the fact that for $y$ binary, $p(1-p)/(y-p) = yp - (1-y)(1-p)$. The fourth equality is by $p^2 - (1-p)^2 = 2p-1$. The fifth is by algebraic simplification. The sixth is by $Y - \mu(A,X)$ being mean zero given $(A,X)$.

For Term II:

\begin{align*}
&\int  h\{P(Y = 1 \mid a,x)\} \frac{d}{d\epsilon}f_\epsilon(a,x)\rvert_{\epsilon = 0} da dx \\ &=  E[h\{P(Y = 1 \mid A,X)\} S(A,X)] \\ &= E[h\{P(Y = 1 \mid A,X)\} S(A,X)] \\ &+ E[h\{P(Y = 1 \mid A,X)\} S(Y \mid A,X)] \\ &= E[h\{P(Y = 1 \mid A,X)\} S(Y, A,X)].
\end{align*}
The second equality is due to $S(Y \mid A,X)$ being mean zero given $(A,X)$.

Summing the two terms, it follows that the uncentered efficient influence function must be $$h(\mu(A,X)) + h'\{\mu(A,X)\} \{Y  - \mu(A,X) \}.$$
\end{proof}

\section{Asymptotic normality}
\subsection{Helper lemmas}
\begin{lemma}
\label{lemma: indicator l2}
Let $A$ and $B$ be random variables and suppose the conditional density $f(a \mid b)$ is bounded above by $C < \infty$ for all $a, b$ in the respective supports with no point masses. Then for functions $h, g$,
\begin{equation*}
E(\mathbbm{1}_{\{g(B) > A > h(B)\}}) \lesssim \norm{h - g}.
\end{equation*}
\end{lemma}
\begin{proof}
We can write
\begin{align*}
E(\mathbbm{1}_{\{h(B) < A < g(B)\}}) &= E_B\left[\int_{h(B)}^{g(B)} f(a \mid B) da \right] \\  &\leq CE_B\left[|h(B)-g(B)|\right] \\  &\lesssim \norm{h - g}.
\end{align*}
The first inequality is by the bounded density and the second is by Cauchy-Schwarz.
\end{proof}

\begin{lemma}
\label{lemma: conditional indicator l2}
Let $A$ and $B$ be random variables and suppose the conditional density $f(a \mid b)$ is bounded above by $C < \infty$ for all $a, b$ in the respective supports with no point masses. Then for functions $h, g$,
\begin{equation*}
\norm{P(A > h(B)| B) - P(A > g(B)| B)} \lesssim \norm{h - g}.
\end{equation*}
\end{lemma}
\begin{proof}
Observe that by definition,
\begin{align*}
P(A > h(B)| B) &= \int_{h(B)}^\infty f(a \mid B) da, \\ P(A > g(B)| B) &= \int_{g(B)}^\infty f(a \mid B) da.
\end{align*}
Taking the difference and absolute value yields
\begin{equation*}
\left|\int_{h(B)}^\infty f(a \mid B) da -  \int_{g(B)}^\infty f(a \mid B) da \right| \leq C |h(B) - g(B)|.
\end{equation*}
Taking the $L_2$ norm on both sides of the above inequality yields 
\begin{align*}
\norm{P(A > h(B)| B) - P(A > g(B)| B)} \lesssim \norm{h - g}, 
\end{align*}
as desired.
\end{proof}

\subsection{Continuous outcome}
\begin{proof}[Proof of Theorem \ref{thm: asymptotic normality continuous}]
Let $\Prob$ and $\Prob_n$ denote expectation and empirical expectation, respectively. We only demonstrate the steps for $\psi_{\textmax}$, as the steps for $\psi_{\textmin}$ are nearly identical. We can decompose
\begin{align*}
\wh{\psi}_{\textmax} &= \Prob_n\phi_{\textmax} (\wh{\eta})  \\ &= \Prob_n\phi_{\textmax} (\wh{\eta}) \pm \Prob_n\phi_{\textmax} ({\eta}) \pm \Prob \phi_{\textmax} ({\eta}) \pm \Prob\phi_{\textmax} (\wh{\eta}) \\ &= \Prob_n\phi_{\textmax} ({\eta}) + \Prob_n\phi_{\textmax} (\wh{\eta}) - \Prob_n\phi_{\textmax} ({\eta}) - 
 \Prob\phi_{\textmax} (\wh{\eta}) + \Prob\phi_{\textmax} ({\eta}) +  \Prob\phi_{\textmax} (\wh{\eta})  -  \Prob \phi_{\textmax} ({\eta})  \\ &= \Prob_n\phi_{\textmax} ({\eta}) + (\Prob_n - \Prob) (\phi_{\textmax} (\wh{\eta}) - \phi_{\textmax} (\eta)) + \Prob(\phi_{\textmax} (\wh{\eta})  -   \phi_{\textmax} ({\eta}) ).
\end{align*}
By the central limit theorem, after the scaling by $\sqrt{n}$, the first term converges to a normal distribution with mean zero and variance matching that in the statement of the theorem. By Lemma 1 from \cite{Kennedy2022}, the second term is $o_p(n^{-1/2})$ since we have assumed that the nuisance estimates $\wh{\eta}$ are consistent for the true nuisances $\eta$, sample/cross-fitting is employed, and $\wh{\eta}$ and ${\eta}$ are uniformly bounded. This brings us to the final term, which is the bias. We can break this into two parts:
\begin{align*}
&T_1 \equiv \Prob[\wh{\mu}'(A,X) - \wh{s}(A,X)\{Y - \wh{\mu}(A,X)\}-{\mu}'(A,X) + {s}(A,X)\{Y - {\mu}(A,X)\}],\\
&T_2 \equiv \gamma \times \Prob[(Y - \wh{M}(A,X))(\mathbbm{1}_{\{Y > \wh{M}( A, X)\}} - \mathbbm{1}_{\{Y < \wh{M}( A, X)\}}) \\ &- (Y - {M}(A,X))(\mathbbm{1}_{\{Y > {M}( A, X)\}} - \mathbbm{1}_{\{Y < {M}( A, X)\}})].
\end{align*}
We can rearrange $T_1$ to be 
\begin{align*}
T_1 &= \Prob[\wh{\mu}'(A,X) - \wh{s}(A,X)\{Y - \wh{\mu}(A,X)\} - \wh{\mu}'(A,X) + {s}(A,X)\{Y - \wh{\mu}(A,X)\}] \\ &+
\Prob[\wh{\mu}'(A,X) - {s}(A,X)\{Y - \wh{\mu}(A,X)\}
-{\mu}'(A,X) + {s}(A,X)\{Y - {\mu}(A,X)\}] \\ &= \Prob[(s(A,X) - \wh{s}(A,X))\{Y - \wh{\mu}(A,X)\}] \\ &+
\Prob[\wh{\mu}'(A,X) + {s}(A,X)\wh{\mu}(A,X)
-{\mu}'(A,X) - {s}(A,X){\mu}(A,X)] \\ &= \Prob[(s(A,X) - \wh{s}(A,X))\{\mu(A,X) - \wh{\mu}(A,X)\}] + 0.
\end{align*}
The last equality is by iterated expectation and integration by parts. Assuming $\wh{s}$ and $\wh{\mu}$ converge at least as fast as $o_p(n^{-1/4})$, by the Cauchy-Schwarz inequality, we get that $T_1 = o_p(n^{-1/2})$. For $T_2$, we can simplify it as 
\begin{align*}
T_2 &= \gamma \times \Prob[(Y - \wh{M}(A,X))(\mathbbm{1}_{\{Y > \wh{M}( A, X)\}} - \mathbbm{1}_{\{Y < \wh{M}( A, X)\}}) \\ &- (Y - {M}(A,X))(\mathbbm{1}_{\{Y > \wh{M}( A, X)\}} - \mathbbm{1}_{\{Y < \wh{M}( A, X)\}})]  \\
&+ \gamma \times \Prob[(Y - {M}(A,X))(\mathbbm{1}_{\{Y > \wh{M}( A, X)\}} - \mathbbm{1}_{\{Y < \wh{M}( A, X)\}})
\\ &- (Y - {M}(A,X))(\mathbbm{1}_{\{Y > {M}( A, X)\}} - \mathbbm{1}_{\{Y < {M}( A, X)\}})] \\ &= \gamma \times \Prob[(M(A,X) - \wh{M}(A,X))(\mathbbm{1}_{\{Y > \wh{M}( A, X)\}}-1/2 + 1/2 - \mathbbm{1}_{\{Y < \wh{M}( A, X)\}})] \\ &+ \gamma \times \Prob[(Y - {M}(A,X))((\mathbbm{1}_{\{Y > \wh{M}( A, X)\}} - \mathbbm{1}_{\{Y < \wh{M}( A, X)\}})- (\mathbbm{1}_{\{Y > {M}( A, X)\}} - \mathbbm{1}_{\{Y < {M}( A, X)\}}))].
\end{align*}

The first term, by iterated expectation, can be rewritten as 
\begin{align*}
\gamma &\times \Prob[(M(A,X) - \wh{M}(A,X))(\mathbbm{1}_{\{Y > \wh{M}( A, X)\}}-1/2 + 1/2 - \mathbbm{1}_{\{Y < \wh{M}( A, X)\}})] \\ &= \gamma \times \Prob[(M(A,X) - \wh{M}(A,X))(P(Y > \wh{M}( A, X) \mid A, X) - P(Y > {M}( A, X) \mid A, X))] \\ &+ \gamma \times \Prob[(M(A,X) - \wh{M}(A,X))(P(Y < {M}( A, X) \mid A, X) - P(Y < \wh{M}( A, X) \mid A, X))].
\end{align*}
This term will $o_p(n^{-1/2})$ by the Cauchy-Schwarz inequality, the assumed rate $\norm{M-\wh{M}} = o_p(n^{-1/4})$, and the fact that the rate of convergence of $\norm{P(Y > \wh{M}( A, X) \mid A, X) - P(Y > {M}( A, X) \mid A, X))}$ and $\norm{P(Y < \wh{M}( A, X) \mid A, X) - P(Y < {M}( A, X) \mid A, X))}$ inherit the rate of convergence of $\norm{M-\wh{M}}$, by Lemma \ref{lemma: conditional indicator l2}.

For the second term, note that for $(\mathbbm{1}_{\{Y > \wh{M}( A, X)\}} - \mathbbm{1}_{\{Y < \wh{M}( A, X)\}})- (\mathbbm{1}_{\{Y > {M}( A, X)\}} - \mathbbm{1}_{\{Y < {M}( A, X)\}})$ to be non-zero, $Y$ must be greater than $\wh{M}$ and less than $M$, or less than $\wh{M}$ and greater than $M$. The probability of such an event (call it $\mathcal{E}$) converges to zero at the same rate as the convergence of $\norm{M-\wh{M}}$, by Lemma \ref{lemma: indicator l2}. Moreover, on the event $\mathcal{E}$, $Y$ lies between $M$ and $\wh{M}$, so $|Y - M(A,X)| \leq |\wh{M}(A,X) -M(A,X)|$ on the event. Then an application of the law of total expectation implies that the second term can be expressed as $$\mathbb{P}[E[Y-M(A,X) \mid \mathcal{E}] \times P(\mathcal{E})] \lesssim \norm{\wh{M}-M}^2 = o_p(n^{-1/2}),$$ by the rate assumptions. In summary, the entire bias term is second order in the estimation error $\norm{M-\wh{M}}$, and consequently is $o_p(n^{-1/2})$. Thus, the central limit theorem holds as stated.
\end{proof}

\subsection{Binary outcome - smoothed functional}
\begin{proof}[Proof of Theorem \ref{thm: asymptotic normality binary}]
Let $\Prob$ and $\Prob_n$ denote expectation and empirical expectation, respectively. As in the proof of Theorem \ref{thm: asymptotic normality continuous}, we only demonstrate the steps for $\psi_{\textmax,h}^B$, as the steps for $\psi_{\textmin,h}^B$ are nearly identical. We omit the $t$ subscript. Similar to the continuous outcome case, we can decompose our estimator as follows:
\begin{align*}
\wh{\psi}_{\textmax,h}^B &= \Prob_n\phi_{\textmax,h}^B (\wh{\eta})  \\ &= \Prob_n\phi_{\textmax,h}^B ({\eta}) + (\Prob_n - \Prob) (\phi_{\textmax,h}^B (\wh{\eta}) - \phi_{\textmax,h}^B (\eta)) + \Prob(\phi_{\textmax,h}^B (\wh{\eta})  -   \phi_{\textmax,h}^B ({\eta}) ).
\end{align*}
By the central limit theorem, after the scaling by $\sqrt{n}$, the first term converges to a normal distribution with mean zero and variance matching that in the statement of the theorem. By Lemma 1 from \cite{Kennedy2022}, the second term is $o_p(n^{-1/2})$ since we have assumed that the nuisance estimates $\wh{\eta}$ are consistent for the true nuisances $\eta$, sample/cross-fitting is employed, and
$\wh{\eta}$ and ${\eta}$ are uniformly bounded. This brings us to the final term, which is the bias. We can break this into two parts:
\begin{align*}
&T_1 \equiv \Prob[\wh{\mu}'(A,X) - \wh{s}(A,X)\{Y - \wh{\mu}(A,X)\}-{\mu}'(A,X) + {s}(A,X)\{Y - {\mu}(A,X)\}]\\
&T_2 \equiv  \gamma \times \Prob[h(\wh{\mu}(A,X)) + h'(\wh{\mu}(A,X))(Y - \wh{\mu}(A,X))] \\ &- \gamma \times \Prob[h({\mu}(A,X)) + h'({\mu}(A,X))(Y - {\mu}(A,X))] \\ &= \gamma \times \Prob[h(\wh{\mu}(A,X)) + h'(\wh{\mu}(A,X))(Y - \wh{\mu}(A,X))] \\ &- \gamma \times \Prob[h({\mu}(A,X))].
\end{align*}
The last equality is because $h'({\mu}(A,X))(Y - {\mu}(A,X))$ is mean zero by iterated expectation. The decomposition of $T_1$ was already done in the proof of Theorem \ref{thm: asymptotic normality continuous}, and is $o_p(n^{-1/2})$ under our assumptions. Next, observe that a simple Taylor expansion of the function $h(p)$ around the point $\wh{p}$ is as follows:
\begin{align*}
h(p) &= h(\wh{p}) + h'(\wh{p})(p-\wh{p}) + O(\{p-\wh{p}\}^2).
\end{align*}
Applying the Taylor expansion taking $\mu$ to be the $p$,
\begin{align*}
1 / \gamma T_2 &= \Prob[h(\wh{\mu}(A,X)) + h'(\wh{\mu}(A,X))(Y - \wh{\mu}(A,X))] \\ &- \Prob[h({\mu}(A,X))] \\ &= \Prob[h(\wh{\mu}(A,X)) + h'(\wh{\mu}(A,X))(Y - \wh{\mu}(A,X))] \\ &- \Prob[h(\wh{\mu}(A,X)) + h'(\wh{\mu}(A,X))(\mu(A,X) - \wh{\mu}(A,X)) \\ &+ O(\{\mu(A,X)-\wh{\mu}(A,X)\})^2] \\ &= \Prob[h'(\wh{\mu}(A,X))(Y - {\mu}(A,X))] + O(\{\mu(A,X)-\wh{\mu}(A,X)\})^2] \\ &= \Prob[O(\{\mu(A,X)-\wh{\mu}(A,X)\})^2] \\ &\lesssim \norm{\mu - \wh{\mu}}^2.
\end{align*}
The second to last equality is because $h'(\wh{\mu}(A,X))(Y - {\mu}(A,X))$ is mean zero by iterated expectation. The last inequality is due to $\wh{\mu}$ and $\mu$ being uniformly bounded. $\norm{\mu - \wh{\mu}}^2$ is $o_p(n^{-1/2})$ by assumption, so $T_2$ is as well.
\end{proof}

\section{Weighted average derivative effects}
\label{sec: weighted ade}
We can also adapt the results to the case where the estimand of interest is a weighted average of the derivative effects. Specifically, consider a known weight function $w(a, x)$, where $w(a, x) \geq 0$ almost surely, and $E[w(A, X)] = 1$. Suppose the estimand of interest is  $E[w(A, X)\partial_a E[Y(A) \mid X, U]]$. This quantity will be identified similarly to its unweighted counterpart, under some additional regularity conditions. 
\begin{assumption}[Regularity - weights]
\label{assumption: regularity - weighted}
The derivative of $w(a, x)$ with respect to $a$ exists, and $f(a \mid x, u) = 0$ implies $w(a, x) = 0$.
\end{assumption}
As \cite{Hines2023OptimallyEffects} discuss, under Assumption \ref{assumption: regularity}, and \ref{assumption: regularity - weighted}, \cite{Powell1989} showed that 
\begin{equation*}
E[w(A,X) \partial_a E[Y \mid A, X, U]] = E[\{-w'(A, X) - w(A, X) s(A \mid X, U)\}Y].    
\end{equation*}
Note that the first component, $E[-w'(A, X) Y]$, is completely identified and is easily estimable when $w(a, x)$ is known. The second component, is simply $E[w(A,X) \times -s(A \mid X, U) Y]$. Since $w$ is known and is nonnegative, the optimization problems defined in \eqref{eq: optimization problem} for both continuous and binary outcomes would have the same optimal solutions for maximizing and minimizing $E[w(A,X) \times -s(A \mid X, U) Y]$, since the optimization problems were solved separately within each $(A = a, X = x)$ strata. The closed-form bounds would be essentially the same, modulo an additional $w(A, X)$ term inside the expectations. Explicitly, these would be 
\begin{equation*}
\begin{aligned}
\psi_{\max} = E[-w'(A, X) Y] + E[-w(A, X)s(A \mid X)Y] + \gamma E[w(A, X)Y (\mathbbm{1}_{\{Y > M( A, X)\}} - \mathbbm{1}_{\{Y < M( A, X)\}})], \\ \psi_{\min} = E[-w'(A, X) Y] + E[-w(A, X)s(A \mid X)Y] - \gamma E[w(A, X)Y (\mathbbm{1}_{\{Y > M( A, X)\}} - \mathbbm{1}_{\{Y < M( A, X)\}})],
\end{aligned}
\end{equation*}
for a continuous outcome and
\begin{equation*}
\begin{aligned}
\psi_{\max, w}^B  = E[-w'(A, X) Y] + E[-w(A, X) s(A\mid X)Y] + \gamma E[w(A, X)\{1/2 - |P(Y = 1 \mid A,X) - 1/2|\}], \\ \psi_{\min, w}^B  = E[-w'(A, X) Y] + E[-w(A, X)s(A\mid X)Y] - \gamma E[w(A, X)\{1/2 - |P(Y = 1 \mid A,X) - 1/2|\}],
\end{aligned}
\end{equation*}
for a binary outcome. Estimation and inference would follow from the same strategy as presented for the unweighted case. The case where the weights are unknown is an interesting direction for future research.
\section{Examining the relationship between model \ref{eqn:sens_model} and Rosenbaum's semiparametric model}
\label{sec: compare marginal and Rosenbaum}
An analogous model for the continuous treatment of Rosenbaum's sensitivity model for binary treatments would look as follows: 
\begin{assumption}[$\gamma$ sensitivity model]
\begin{equation}
\label{eqn:sens_model 1}
    \exp(-\gamma(|a - a'|)) \leq \frac{f(a' \mid x, u)f(a\mid x, u')}{f(a\mid x, u)f(a' \mid x, u')}\leq  \exp(\gamma(|a - a'|)) \ \forall x, a, a', u, u'.
\end{equation}
\end{assumption}
The model restricts the odds ratio of the generalized propensity scores at any two dose levels and any two values of the unmeasured confounder. This model directly generalizes Rosenbaum's sensitivity model for binary treatments, i.e. when $a$ and $a'$ are replaced by 0 and 1 \citep{rosenbaum_obs}. The key difference between models \eqref{eqn:sens_model} and \eqref{eqn:sens_model 1} is that the latter model compares quantities conditioning on any two different values of the unmeasured confounder, $u$ and $u'$, whereas the former compares a quantity conditioning on any $u$ with an analogous quantity but with $u$ marginalized out. We now demonstrate that the semiparametric model introduced in \citep{rosenbaum1989sensitivity} implies model \eqref{eqn:sens_model 1}.
\begin{example}[Rosenbaum model for continuous doses]
\label{example: rosenbaum model}
Suppose that $$f(a \mid u, x) = \zeta(x,u)\eta(a,x) \exp(\gamma a u),$$ where $\eta(a,x)$ is an arbitrary function, and $\zeta(x,u)$ is a normalizing constant that ensures $f(a \mid u, x)$ integrates to 1. 
\end{example}
We will now show that if $u \in [0,1]$, then sensitivity model \eqref{eqn:sens_model 1} holds at $\gamma$. To see this, directly plugging into Equation \eqref{eqn:sens_model 1}, we get 
\begin{align*}
    \frac{f(a' \mid x, u)f(a\mid x, u')}{f(a\mid x, u)f(a' \mid x, u')} &= \frac{\zeta(x,u)\eta(a',x) \exp(\gamma a' u)\zeta(x,u')\eta(a,x) \exp(\gamma a u')}{\zeta(x,u)\eta(a,x) \exp(\gamma a u)\zeta(x,u')\eta(a',x) \exp(\gamma a' u')} 
\\ &= \frac{ \exp(\gamma a' u) \exp(\gamma a u')}{\exp(\gamma a u)\exp(\gamma a' u')} \\ &= \exp(\gamma(a'-a)(u-u')) \\ &\in [ \exp(-\gamma(|a - a'|)),  \exp(\gamma(|a - a'|))].
\end{align*}
The model from Example \ref{example: rosenbaum model} covers a range of familiar settings. For example, the situation where $A \mid X, U \sim N(g(X) + \gamma U, 1)$, where $g$ is an arbitrary function, falls within the model. Another is when $A \mid X, U \sim \text{Gamma}(a, g(X) - \gamma U)$, for a fixed constant $a$ and an arbitrary function $g$. There are many other exponential family models for which Example \ref{example: rosenbaum model} applies. It is also compatible with binary and ordinal treatments.

The next result elucidates a connection between models \eqref{eqn:sens_model 1} and \eqref{eqn:sens_model}; a similar result for the binary case was derived in \cite{Zhao2019SensitivityBootstrap}. The following lemma is a slight generalization of Proposition 3 of \cite{Zhao2019SensitivityBootstrap}.
\begin{lemma}
\label{lemma: marginal gamma / 2 implies gamma}
If model \eqref{eqn:sens_model} holds at $\gamma / 2$, then model \eqref{eqn:sens_model 1} holds at $\gamma$. Also, if model \eqref{eqn:sens_model 1} holds at $\gamma$, then model \eqref{eqn:sens_model} holds at $\gamma$.
\end{lemma}
\begin{proof}[Proof of Lemma \ref{lemma: marginal gamma / 2 implies gamma}]
Fix any $u, u'$, $x$, and $a, a'$ and suppose model \eqref{eqn:sens_model} holds at $\gamma / 2$. Then taking logs, we get 
\begin{equation*}
-\gamma/2(|a - a'|) \leq \log(f(a' \mid x, u)) + \log(f(a \mid x)) - \log(f(a \mid x, u)) - \log(f(a' \mid x)) \leq \gamma/2(|a - a'|),
\end{equation*}
and 
\begin{equation*}
-\gamma / 2(|a - a'|) \leq -\log(f(a' \mid x, u')) - \log(f(a \mid x)) + \log(f(a \mid x, u')) + \log(f(a' \mid x)) \leq \gamma / 2(|a - a'|).
\end{equation*}
Adding these inequality chains, we get 
\begin{equation*}
-\gamma(|a - a'|) \leq \log(f(a' \mid x, u)) + \log(f(a \mid x, u')) - \log(f(a \mid x, u)) - \log(f(a' \mid x, u')) \leq \gamma(|a - a'|).
\end{equation*}
After taking exponents, we get that model \eqref{eqn:sens_model 1} holds at $\gamma$. 
For the second implication, suppose that model \eqref{eqn:sens_model 1} holds at $\gamma$. Then
\begin{equation*}
\exp(-\gamma(|a - a'|))f(a' \mid x, u') f(a \mid x, u) \leq f(a' \mid x, u)  f(a \mid x, u') \leq \exp(\gamma(|a - a'|))f(a' \mid x, u')f(a \mid x, u).
\end{equation*}
We can simply integrate the chain of inequalities with respect to $f(u' \mid x)$ over the support of $U$, to get 
\begin{equation*}
\exp(-\gamma(|a - a'|))f(a' \mid x) f(a \mid x, u) \leq f(a' \mid x, u)  f(a \mid x) \leq \exp(\gamma(|a - a'|))f(a' \mid x)f(a \mid x, u),
\end{equation*}
which is equivalent to model \eqref{eqn:sens_model}.

\end{proof}

The second half of the result demonstrates that assuming \eqref{eqn:sens_model 1} holds at $\gamma$ is stronger than assuming \eqref{eqn:sens_model} holds at $\gamma$. Meanwhile, the first half of the result implies that the interpretation of the $\gamma$'s from the differing models does not exceed a factor of 2. We refer the reader to \cite{dalal2025partial} for more detailed comparisons between the Tan (marginal) and Rosenbaum sensitivity models.

\section{Alternative transforms in the sensitivity model}
\label{sec: alternative transform}
Recall that in model \eqref{eqn:sens_model}, the odds ratio $\frac{f(a' \mid x, u)f(a\mid x)}{f(a\mid x, u)f(a' \mid x)}$ is bounded between $\exp(-\gamma(|a-a'|))$ and $\exp(\gamma(|a-a'|))$. One might instead consider a model that uses a function other than exponential.
\begin{assumption}[$g$ marginal sensitivity model]
\begin{equation}
\label{eqn: g model}
g(-\gamma(|a-a'|)) \leq \frac{f(a' \mid x, u)f(a\mid x)}{f(a\mid x, u)f(a' \mid x)} \leq g(\gamma(|a-a'|)),
\end{equation}
for some smooth, nonnegative, strictly increasing function $g$ such that $g(0) = 1$.
\end{assumption}
Under such a model, we can follow the logic of Lemma \ref{lemma: model implication on latent score} to derive a restriction on the latent score $s(a \mid x, u)$. 

\begin{lemma}
\label{lemma: g model implication on latent score}
Under sensitivity model \eqref{eqn: g model}, 
\begin{equation}
s(a \mid x) +\lim_{h \to 0}\frac{\log(g(-\gamma h))}{h}\leq s(a \mid x, u) \leq s(a \mid x) +\lim_{h \to 0}\frac{\log(g(\gamma h))}{h}, \ \forall a, x, u.
\end{equation}
\end{lemma}
\begin{proof}
We can apply logs to \eqref{eqn: g model} to get
\begin{equation*}
\log(g(-\gamma |a - a'|)) \leq \log(f(a \mid x)) - 
 \log(f(a' \mid x)) - 
 \{ \log(f(a \mid x,u )) - \log(f(a' \mid x, u))\} \leq \log(g(\gamma |a - a'|)) .
\end{equation*}
Next, recall that $s(a \mid x,u) = {\partial_a} \log(f(a\mid x,u)) = \lim_{h \to 0} \frac{\log(f(a + h\mid x,u)) - \log(f(a\mid x,u))}{h}$. By the above equation, plugging in $a$ for $a$ and $a+h$ for $a'$, adding the $\log(f(a + h\mid x)) - \log(f(a\mid x))$ term everywhere in the inequality and dividing by $h$, we get 
\begin{equation*}
\begin{aligned}
    \frac{\log(f(a + h\mid x)) - \log(f(a\mid x))+\log(g(-\gamma h))}{h} &\leq \frac{\log(f(a + h\mid x,u)) - \log(f(a\mid x,u))}{h} \\ &\leq \frac{\log(f(a + h\mid x)) - \log(f(a\mid x))+\log(g(\gamma h))}{h}, 
\end{aligned}
\end{equation*}
which immediately implies the result after taking the limit $h \to 0$ everywhere.
\end{proof}
One could then solve the optimization problem \eqref{eq: optimization problem}, except replacing constraint $s(a \mid x) -\gamma \leq s(a \mid x, u) \leq s(a \mid x) + \gamma$ with $s(a \mid x) +\lim_{h \to 0}\frac{\log(g(-\gamma h))}{h}\leq s(a \mid x, u) \leq s(a \mid x) +\lim_{h \to 0}\frac{\log(g(\gamma h))}{h}$.

\end{document}